\documentclass[11pt,twoside,letterpaper,notitlepage]{article}
\usepackage[noblocks]{authblk}

\usepackage{ifthen}

\newboolean{conferenceversion}
\setboolean{conferenceversion}{false}

\usepackage[a4paper,margin=2.5cm]{geometry}
\usepackage{times}
    
\usepackage[english]{babel}

\usepackage{fancyhdr}

\usepackage{graphics}

\usepackage[sf,SF]{subfigure}

\usepackage{bussproofs}

\usepackage{ifthen}

\usepackage{ifthen}

\newboolean{maybeSTOC}
\newboolean{maybeFOCS}
\newboolean{maybeElsevier}
\newboolean{maybeNOW}
\newboolean{maybeLMCS}
\newboolean{maybePoster}
\newboolean{maybeSIAM}
\newboolean{maybeIEEE}
\newboolean{maybeICS}
\newboolean{maybeLNCS}
\newboolean{maybeACM}
\newboolean{maybeArticle}
\newboolean{maybeReport}
\newboolean{maybeThesis}

\ifthenelse{\isundefined{\affaddr}}        
{\setboolean{maybeSTOC}{false}}
{\setboolean{maybeSTOC}{true}}

\ifthenelse{\not \isundefined{\affiliation}
  \and \not \isundefined{\Section}}        
{\setboolean{maybeFOCS}{true}}
{\setboolean{maybeFOCS}{false}}

\ifthenelse{\not \isundefined{\tnotetext}
  \and \not \isundefined{\ead}}        
{\setboolean{maybeElsevier}{true}}
{\setboolean{maybeElsevier}{false}}

\ifthenelse{\not \isundefined{\journal} \and 
  \not \isundefined{\volume} \and 
  \not \isundefined{\issue} \and 
  \not \isundefined{\copyrightowner} \and 
  \not \isundefined{\articletitle} \and 
  \not \isundefined{\pubyear}}
{\setboolean{maybeNOW}{true}}
{\setboolean{maybeNOW}{false}}

\ifthenelse{\not \isundefined{\keywords} \and
  \not \isundefined{\subjclass} \and
  \not \isundefined{\titlecomment} \and
  \not \isundefined{\revisionname} \and
  \not \isundefined{\lmcsheading}}
{\setboolean{maybeLMCS}{true}}
{\setboolean{maybeLMCS}{false}}

\ifthenelse{\not \isundefined{\IEEEtransversionmajor} \and 
  \not \isundefined{\IEEEtransversionminor} \and 
  \not \isundefined{\IEEEitemize} \and 
  \not \isundefined{\IEEEenumerate} \and 
  \not \isundefined{\IEEEdescription}}
{\setboolean{maybeIEEE}{true}}
{\setboolean{maybeIEEE}{false}}

\ifthenelse{\isundefined{\sicomp}}
{\setboolean{maybeSIAM}{false}}
{\setboolean{maybeSIAM}{true}}

\ifthenelse{\not \isundefined{\address} \and 
  \not \isundefined{\email} \and 
  \not \isundefined{\keywords} \and 
  \not \isundefined{\sanhao} \and 
  \not \isundefined{\wuhao}}
{\setboolean{maybeICS}{true}}
{\setboolean{maybeICS}{false}}

\ifthenelse{\not \isundefined{\institute} \and 
  \not \isundefined{\institutename} \and 
  \not \isundefined{\email} \and 
  \not \isundefined{\fnmsep}}
{\setboolean{maybeLNCS}{true}}
{\setboolean{maybeLNCS}{false}}

\ifthenelse{\not \isundefined{\acmVolume} \and 
  \not \isundefined{\acmNumber} \and 
  \not \isundefined{\acmArticle} \and 
  \not \isundefined{\acmYear} \and 
  \not \isundefined{\acmMonth}}
{\setboolean{maybeACM}{true}}
{\setboolean{maybeACM}{false}}

\ifthenelse{\isundefined{\conference}}
{\setboolean{maybePoster}{false}}
{\setboolean{maybePoster}{true}}

\ifthenelse{\isundefined{\chapter}}
{\setboolean{maybeArticle}{true}}
{\setboolean{maybeArticle}{false}}

\ifthenelse{\not \isundefined{\examen} 
  \and \not \isundefined{\disputationsdatum} 
  \and \not \isundefined{\disputationslokal}}   
  {\setboolean{maybeThesis}{true}}
  {\setboolean{maybeThesis}{false}}
           
\ifthenelse{\boolean{maybeArticle} \or \boolean{maybeThesis}
  \or \boolean{maybeSTOC} \or \boolean{maybeFOCS}
  \or \boolean{maybeSIAM} \or \boolean{maybeIEEE}
  \or \boolean{maybeICS} \or \boolean{maybePoster}}
{\setboolean{maybeReport}{false}}
{\setboolean{maybeReport}{true}}

\usepackage[utf8]{inputenc}

\usepackage{amsmath}
\usepackage{amssymb}
\usepackage{amsfonts}

\usepackage{mathtools}

\ifthenelse
{\boolean{maybeElsevier}}
{}
{\usepackage{varioref}}

\usepackage{xspace}

\ifthenelse
{\boolean{maybeSTOC} \or \boolean{maybeSIAM} \or \boolean{maybeLMCS}
 \or \boolean{maybeNOW} \or \boolean{maybeLNCS} \or \boolean{maybeACM}}
{}
{\usepackage{local-amsthm}}

\ifthenelse
{\boolean{maybeElsevier}    \or \boolean{maybeFOCS} \or \boolean{maybeSTOC}
  \or \boolean{maybePoster} \or \boolean{maybeSIAM} \or \boolean{maybeLMCS}
  \or \boolean{maybeIEEE}   \or \boolean{maybeNOW}  \or \boolean{maybeICS}
  \or \boolean{maybeThesis} \or \boolean{maybeLNCS} \or \boolean{maybeACM}}
{}
{\usepackage{nada-typography}}

\usepackage{ifthen}
\usepackage{xspace}
\ifthenelse
{\boolean{maybeElsevier}}
{}
{\usepackage{varioref}}

\DeclareMathAlphabet{\mathsfsl}{OT1}{cmss}{m}{sl}

\newcommand{\eqperiod}{\enspace .}

\newcommand{\formatfunctiontoset}[1]{\mathit{#1}}

\newcommand{\introduceterm}[1]{{\emph{#1}}}

\newcommand{\ie}{i.e.,\ }

\ifthenelse{\boolean{maybeIEEE}}{}{}

\newcommand{\bigomega}[1]{\Omega ( #1 )}

\newcommand{\problemlanguageformat}[1]{\textsc{#1}\xspace}

\newcommand{\MINIMALUNSATISFIABILITY}%
  {\problemlanguageformat{minimal unsatisfiability}}

\newcommand{\complclassformat}[1]{\textrm{\upshape{\textsf{#1}}}\xspace}

\newcommand{\cocomplclass}[1]%
        {\mbox{\complclassformat{co}-\complclassformat{#1}}\xspace}

\newcommand{\refsec}[1]{Section~\ref{#1}}

\newcommand{\reffig}[1]{Figure~\ref{#1}}

\newcommand{\reftwofigs}[2]{Figures~\ref{#1} and~\ref{#2}}

\newcommand{\refth}[1]{Theorem~\ref{#1}}

\newcommand{\reflem}[1]{Lemma~\ref{#1}}

\newcommand{\refpr}[1]{Proposition~\ref{#1}}

\newcommand{\refcor}[1]{Corollary~\ref{#1}}

\newcommand{\refdef}[1]{Definition~\ref{#1}}

\newcommand{\refobs}[1]{Observation~\ref{#1}}

\ifthenelse
{\isundefined{\refeq}}
{\newcommand{\refeq}[1]{\eqref{#1}}}
{\renewcommand{\refeq}[1]{\eqref{#1}}}

\newcommand{\MAXOFSET}[3][:]{\max \left\{ #2 #1 #3 \right\}}
\newcommand{\MINOFSET}[3][:]{\min \left\{ #2 #1 #3 \right\}}

\newcommand{\fieldstd}{\mathbb{F}}

\DeclareMathOperator{\Expop}{E}

\ifthenelse{\boolean{maybeLMCS}}
{}
{}

\newcommand{\twincommandJN}[6]%
    {#1#2#3\vphantom{#2#5}\mspace{-2.25mu}#4.#5#6}

\newcommand{\CondExp}[2]%
    {\Expop\twincommandJN{\bigl[}{#1}{\bigl|}{\bigr}{\,#2}{\bigr]}}
\newcommand{\CONDEXP}[2]%
     {\Expop\twincommandJN{\left[}{#1}{\left|}{\right}{\,#2}{\right]}}

\newcommand{\CondProb}[3][]%
    {\Pr_{#1}\twincommandJN{\bigl[}{#2}{\bigl|}{\bigr}{\,#3}{\bigr]}}
\newcommand{\CONDPROB}[3][]%
    {\Pr_{#1}\twincommandJN{\left[}{#2}{\left|}{\right}{\,#3}{\right]}}

\newcommand{\isdistras}[2]{\ensuremath{#1} \sim \ensuremath{#2}}
\newcommand{\funcdescr}[3]{\ensuremath{ #1\colon #2 \to #3}}

\newcommand{\vertices}[1]{V( #1 )}

\newcommand{\setcompact}[1]{{\ensuremath{\bigl\{ #1 \bigr\}}}}
\newcommand{\setsmall}[1]{{\ensuremath{\{ #1 \}}}}
\newcommand{\setdescrcompact}[3][\mid]{{\setcompact{ #2 #1 #3 }}}

\newcommand{\set}[1]{\{ #1 \}}

\newcommand{\setdescr}[3][\mid]{\setsmall{ #2 #1 #3 }}
\newcommand{\Setdescr}[3][|]%
     {\twincommandJN{\bigl\{}{#2}{\bigl#1}{\bigr}{\,#3}{\bigr\}}}
\newcommand{\SETDESCR}[3][|]%
     {\twincommandJN{\left\{}{#2}{\left#1}{\right}{\,#3}{\right\}}}

\newcommand{\Setdescrbrackets}[3][|]%
     {\twincommandJN{\bigl[}{#2}{\bigl#1}{\bigr}{\,#3}{\bigr]}}
\newcommand{\SETDESCRBRACKETS}[3][|]%
     {\twincommandJN{\left[}{#2}{\left#1}{\right}{\,#3}{\right]}}

\newcommand{\setsize}[1]{\lvert#1\rvert}

\newcommand{\union}{\cup}
\newcommand{\Union}{\bigcup}

\newcommand{\unionSP}{\, \union \, }

\newcommand{\DisjointunionInText}%
    {{\smash{\overset{\mbox{\boldmath{.}}}{\bigcup}}}\vphantom{\bigcup}}
\newcommand{\intnfirst}[1]{[{#1}]}

\newcommand{\Lor}{\bigvee}
\newcommand{\Land}{\bigwedge}

\newcommand{\Lornodisplay}{{\textstyle \bigvee}}

\newcommand{\olnot}[1]{\overline{#1}}
\newcommand{\stdnot}[1]{\olnot{#1}}
\newcommand{\cnfform}{\cnfshort for\-mu\-la\xspace}

\newcommand{\cnfshort}{CNF\xspace}

\newcommand{\xcnfform}[1]{\mbox{\ensuremath{#1}-}\cnfform}
\newcommand{\kcnfform}{\xcnfform{\clwidth}}
\newcommand{\xclause}[1]{\mbox{\ensuremath{#1}-clause}\xspace}

\newcommand{\nvar}{n}
\newcommand{\nclause}{m}
\newcommand{\clwidth}{k}

\newcommand{\randkcnfnclwrepl}[3][\clwidth]%
        {\ensuremath{\mathcal{F}^{#2, #3}_{#1}}}
\newcommand{\randkcnfnclwreplstd}% 
        {\randkcnfnclwrepl{\clwidth}{\nvar}{\nclause}}

\newcommand{\israndkcnfnclwrepl}[4]%
  {\isdistras{#1}{\randkcnfnclwrepl[#2]{#3}{#4}}}

\newcommand{\randkcnfprobcl}[3]%
        {\ensuremath{\mathcal{F}^{#2}_{#1} \bigl(#3 \bigr)}}

\newcommand{\pcfor}[4][to]{for #2 := #3 #1 #4 do}
\newcommand{\pcformath}[4][to]%
    {\pcfor[#1]{\ensuremath{#2}}{\ensuremath{#3}}{\ensuremath{#4}}}

\newcommand{\pcassigncompact}[2]{#1 := #2}
\newcommand{\pcassignmathcompact}[2]%
        {\pcassigncompact{\ensuremath{#1}}{\ensuremath{#2}}}

\newcommand{\inductionformat}[1]{\textit{#1}}

\newcommand{\BASE}[1][]
        {\inductionformat
                {%
                        \ifthenelse{\equal{#1}{}}%
                                {Base case: }%
                                {Base case (#1):}%
                }%
        }

\ifthenelse{\isundefined{\DONOTLOADHYPERREFPACKAGE}
  \and \not \boolean{maybeSTOC}     \and \not \boolean{maybeFOCS}
  \and \not \boolean{maybeElsevier} \and \not \boolean{maybePoster}
  \and \not \boolean{maybeSIAM}     \and \not \boolean{maybeACM}
  \and \not \boolean{maybeIEEE}     \and \not \boolean{maybeNOW}
  \and \not \boolean{maybeICS}      \and \not \boolean{maybeThesis}
  \and \not \boolean{maybeLNCS}}
{\usepackage[colorlinks=true,citecolor=blue]{hyperref}}
{}

\ifthenelse{\boolean{maybeSTOC}}
   {\input{definitions-STOC.tex}}
   {\ifthenelse{\boolean{maybeSIAM}}
     {\input{definitions-SIAM.tex}}
     {\ifthenelse{\boolean{maybeNOW}}
       {\input{definitions-NOW.tex}}
       {\ifthenelse{\boolean{maybeLMCS}}
         {\input{definitions-lmcs.tex}}
         {\ifthenelse{\boolean{maybeIEEE}}
           {\input{definitions-IEEE.tex}}
           {\ifthenelse{\boolean{maybeICS}}
             {\input{definitions-ICS.tex}}
             {\ifthenelse{\boolean{maybeLNCS}}
               {\input{definitions-LNCS.tex}}
               {\ifthenelse{\boolean{maybeACM}}
                 {\input{definitions-ACM.tex}}
                 {\ifthenelse{\boolean{maybeArticle} \or \boolean{maybeElsevier}}
                   {%
\newtheorem{standardlocalcounter}{Dummy}[section]

\theoremstyle{plain}    

\newtheorem{theorem}[standardlocalcounter]{Theorem}
\newtheorem{lemma}[standardlocalcounter]{Lemma}
\newtheorem{proposition}[standardlocalcounter]{Proposition}
\newtheorem{corollary}[standardlocalcounter]{Corollary}
\newtheorem{observation}[standardlocalcounter]{Observation}

\theoremstyle{definition}

\newtheorem{definition}[standardlocalcounter]{Definition}

\theoremstyle{remark}

\newtheoremstyle{meta}% name
  {3pt}%      Space above
  {3pt}%      Space below
  {\scshape \small }%         Body font
  {}%         Indent amount (empty = no indent, \parindent = para indent)
  {\scshape \small }% Thm head font
  {:}%        Punctuation after thm head
  { }%     Space after thm head: " " = normal interword space;
  {}%         Thm head spec (can be left empty, meaning `normal')

\theoremstyle{meta}

\newtheoremstyle{questions}% name
  {3pt}%      Space above
  {3pt}%      Space below
  {\sffamily \slshape}%         Body font
  {}%         Indent amount (empty = no indent, \parindent = para indent)
  {\bfseries \sffamily \slshape}% Thm head font
  {:}%        Punctuation after thm head
  { }%     Space after thm head: " " = normal interword space;
  {}%         Thm head spec (can be left empty, meaning `normal')

\theoremstyle{questions}

}
                   {\input{definitions-report.tex}}}}}}}}}}

\ifthenelse{\boolean{maybeThesis}}
{}
{\usepackage{ifpdf}}

\usepackage{graphicx}

\ifpdf
\fi

\ifthenelse
{\boolean{maybeSTOC} \or \boolean{maybeThesis} \or \boolean{maybeLNCS}}
{}
{
  \setcounter{secnumdepth}{3}
  \setcounter{tocdepth}{3}
}

\newcount\hours
\newcount\minutes
\def\SetTime{\hours=\time
\global\divide\hours by 60
\minutes=\hours
\multiply\minutes by 60
\advance\minutes by-\time
\global\multiply\minutes by-1 }
\SetTime
\def\now{\number\hours:\ifnum\minutes<10 0\fi\number\minutes}

\newcommand{\formuladots}{\cdots}

\newcommand{\impl}{\vDash}
\newcommand{\nimpl}{\nvDash}

\newcommand{\proofsystemformat}[1]{\ensuremath{\mathfrak{#1}}}

\newcommand{\proofstd}{\ensuremath{\pi}}

\newcommand{\resknot}[1][k]{\proofsystemformat{R}({#1})}

\newcommand{\derivof}[4][\derives]
        {{\ensuremath{{#2} : {#3} \, {#1}\, {#4}}}}

\newcommand{\refof}[2]{\derivof{#1}{#2}{\falsenum}}

\newcommand{\deriveswithall}%
        {\vdash_{\!\!\!{\scriptscriptstyle \forall}}} 
\newcommand{\notderiveswithall}%
        {\nvdash_{\!\!\!{\scriptscriptstyle \forall}}} 
\newcommand{\clcfgtransitioncrammed}[2]%
        {\ensuremath{#1 \!\rightsquigarrow\! #2}}

\newcommand{\tvastd}{{\ensuremath{\alpha}}}

\newcommand{\fstd}{{\ensuremath{F}}}

\newcommand{\emptycl}{0}

\newcommand{\varx}{\ensuremath{x}}

\newcommand{\lita}{\ensuremath{a}}

\newcommand{\cla}{\ensuremath{A}}

\newcommand{\clc}{\ensuremath{C}}
\newcommand{\cld}{\ensuremath{D}}

\newcommand{\clausesetformat}[1]{\ensuremath{\mathbb{#1}}}
\newcommand{\clsc}{\clausesetformat{C}}
\newcommand{\clsd}{\clausesetformat{D}}

\newcommand{\termt}{\ensuremath{T}}
\newcommand{\pcpolyp}{P}

\newcommand{\pcmonm}{m}

\newcommand{\setsofvarsorlitlarge}[2]%
        {\mathit{#1}\left({#2}\right)}
\newcommand{\setsofvarsorlit}[2]%
        {\mathit{#1}({#2})}
\newcommand{\setsofvarsorlitcompact}[2]%
        {\mathit{#1}\bigl({#2}\bigr)}

\newcommand{\setsofvarsorlitsup}[3]%
        {\mathit{#1}^{#2}({#3})}
\newcommand{\setsofvarsorlitsuplarge}[3]%
        {\mathit{#1}^{#2}\left({#3}\right)}
\newcommand{\setsofvarsorlitsupcompact}[3]%
        {\mathit{#1}^{#2}\bigl({#3}\bigr)}

\newcommand{\derivabbrev}[2]{\bigl( #1 \vdash #2 \bigr)}
\newcommand{\derivabbrevsmall}[2]{( #1 \vdash #2 )}
\newcommand{\derivabbrevcompact}[2]{\bigl( #1 \vdash #2 \bigr)}

\newcommand{\refutabbrevsmall}[1]{\derivabbrevsmall{#1}{\falsenum}}
\newcommand{\refutabbrevcompact}[1]{\derivabbrevcompact{#1}{\falsenum}}

\newcommand{\genericformsmall}[2]{\mathit{#1}( #2 )}

\newcommand{\genericrefsmall}[3]%
    {{\mathit{#1}}_{#2}\refutabbrevsmall{#3}}
\newcommand{\genericrefcompact}[3]%
    {{\mathit{#1}}_{#2}\refutabbrevcompact{#3}}
\newcommand{\genericderiv}[4]%
    {{\mathit{#1}}_{#2}\derivabbrev{#3}{#4}}
\newcommand{\genericderivsmall}[4]%
    {{\mathit{#1}}_{#2}\derivabbrevsmall{#3}{#4}}
\newcommand{\genericderivcompact}[4]%
    {{\mathit{#1}}_{#2}\derivabbrevcompact{#3}{#4}}
\newcommand{\generictaut}[3]%
    {{\mathit{#1}}_{#2}\derivabbrev{}{#3}}
\newcommand{\generictautcompact}[3]%
    {{\mathit{#1}}_{#2}\derivabbrevcompact{}{#3}}
\newcommand{\generictautsmall}[3]%
    {{\mathit{#1}}_{#2}\derivabbrevsmall{}{#3}}
\newcommand{\lengthofarg}[1]{\genericformsmall{L}{#1}}

\newcommand{\lengthref}[2][]{\genericrefsmall{L}{#1}{#2}}

\newcommand{\widthofarg}[2][]{\genericformsmall{W_{#1}}{#2}}

\newcommand{\widthref}[2][]{\genericrefsmall{W}{#1}{#2}}

\newcommand{\clspaceof}[2][]{\genericformsmall{Sp_{#1}}{#2}}

\newcommand{\clspaceref}[2][]{\genericrefsmall{Sp}{#1}{#2}}

\newcommand{\formulaformat}[1]{\ensuremath{\mathit{#1}}}
\renewcommand{\formulaformat}[1]{\mathit{#1}}

\newcommand{\transitionarrow}{\rightsquigarrow}
\newcommand{\pebcfgtransition}[2]%
    {\ensuremath{#1 \transitionarrow #2}}
\newcommand{\pebcfgtransitionsqueeze}[2]%
    {#1 \! \transitionarrow \! #2}

\newcommand{\formatpebblingprice}[1]{\text{\textsl{\textsf{#1}}}}

\newcommand{\Pebblingprice}[1]%
    {\formatpebblingprice{Peb}\bigl(#1\bigr)}
\newcommand{\pebblingpricecompact}[1]%   OBSOLETE
    {\formatpebblingprice{Peb}\bigl(#1\bigr)}

\newcommand{\Bwpebblingprice}[1]%
    {\formatpebblingprice{BW-Peb}\bigl(#1\bigr)}
\newcommand{\bwpebblingpricecompact}[1]%   OBSOLETE
    {\formatpebblingprice{BW-Peb}\bigl(#1\bigr)}

\newcommand{\pebpersistentsymbol}{\bullet}
\newcommand{\pebvisitingsymbol}{\emptyset}

\newcommand{\bwpebpricepersistent}[1]%
    {\formatpebblingprice{BW-Peb}^{\pebpersistentsymbol}(#1)}
\newcommand{\Bwpebpricepersistent}[1]%
    {\formatpebblingprice{BW-Peb}^{\pebpersistentsymbol}\bigl(#1\bigr)}
\newcommand{\bwpebpricevisiting}[1]%
    {\formatpebblingprice{BW-Peb}^{\pebvisitingsymbol}(#1)}
\newcommand{\Bwpebpricevisiting}[1]%
    {\formatpebblingprice{BW-Peb}^{\pebvisitingsymbol}\bigl(#1\bigr)}

\newcommand{\pebpricepersistent}[1]%
    {\formatpebblingprice{Peb}^{\pebpersistentsymbol}(#1)}
\newcommand{\Pebpricepersistent}[1]%
    {\formatpebblingprice{Peb}^{\pebpersistentsymbol}\bigl(#1\bigr)}
\newcommand{\pebpricevisiting}[1]%
    {\formatpebblingprice{Peb}^{\pebvisitingsymbol}(#1)}
\newcommand{\Pebpricevisiting}[1]%
    {\formatpebblingprice{Peb}^{\pebvisitingsymbol}\bigl(#1\bigr)}

\newcommand{\bwpebblingpriceempty}[1]%
    {\formatpebblingprice{BW-Peb}^{\pebvisitingsymbol}(#1)}
\newcommand{\bwpebblingpriceemptycompact}[1]%
    {\formatpebblingprice{BW-Peb}^{\pebvisitingsymbol}\bigl(#1\bigr)}

\newcommand{\stoptime}{\tau}
\newcommand{\pebcontr}[2][G]{\ensuremath{\formulaformat{Peb}^{#2}_{#1}}}

\newcommand{\pebdeg}{\ensuremath{d}}

\newcommand{\pebaxcompact}[2]%
        [\pebdeg]{\ensuremath{\formulaformat{Ax}^{#1} \bigl(#2 \bigr)}}

\newcommand{\pqrxvar}[6]%
    {\ensuremath{\stdnot{\varx({#1})}_{#2} \lor \stdnot{\varx({#3})}_{#4} \lor % 
    \sourceclausexvar[#6]{#5}}}

\newcommand{\pqr}[6]%
    {\ensuremath{\stdnot{#1}_{#2} \lor \stdnot{#3}_{#4} \lor % 
    \sourceclausenodisplay[#6]{#5}}}
\newcommand{\pqrstd}{\pqr{p}{i}{q}{j}{r}{l}}
\newcommand{\pqrall}[6]%
        {\setdescrcompact
        {\pqr{#1}{#2}{#3}{#4}{#5}{#6}}{#2,#4 \in \intnfirst{\pebdeg}}}
\newcommand{\pqrallstd}%
        {\setdescrcompact{\pqrstd}{i,j \in \intnfirst{\pebdeg}}}

\newcommand{\sourceclausexvar}[2][n]%
        {\Lor_{#1 = 1}^{\pebdeg} \varx({#2})_{#1}}
\newcommand{\subsourceclausexvar}[3][n]%
        {\Lor_{#1 = {#2}}^{\pebdeg} \varx({#3})_{#1}}

\newcommand{\sourceclausexvarnodisplay}[2][n]%
        {\textstyle \Lor_{#1 = 1}^{\pebdeg} \varx({#2})_{#1}}

\newcommand{\sourceclausenodisplay}[2][n]%
        {\textstyle \Lor_{#1 = 1}^{\pebdeg} #2_{#1}}

\newcommand{\tseitinnot}[2]{\ensuremath{\formulaformat{T}_{#1,#2}}}
\newcommand{\tseitinparity}[2]{\ensuremath{\formulaformat{PARITY}_{#1,#2}}}

\newcommand{\relativisation}[1]%
    {\ensuremath{\formulaformat{Rel}\bigl(#1 \bigr)}}

\newcommand{\extPHPnot}[2]
    {\ensuremath{\extendedversion{\formulaformat{PHP}}^{#1}_{#2}}}

\newcommand{\extendedversion}[1]{\widetilde{#1}}

\usepackage{stmaryrd} 

\newcommand{\formatfunctiontosubconfiguration}[1]{\mathsf{#1}}
\newcommand{\formatfunctiontomulti}[1]{\mathcal{#1}}

\DeclareMathOperator{\dummystar}{*}
\newcommand{\pebblingcontrNT}[2][G]%
 {\ensuremath{\dummystar\!\!\formulaformat{Peb}^{#2}_{#1}}}

\newcommand{\somenodetrueclausedeg}[2]{\formulaformat{All}_{#1}^{+}({#2})}

\newcommand{\slashedstrickenletter}[1]{{\backslash\mkern-9mu #1}}
            
\newcommand{\strikethroughcommand}[1]{\slashedstrickenletter{#1}}

\newcommand{\abovevertices}[2][G]%
    {{#1}_{#2}^{\hspace{-0.2 pt}\triangledown}}
\newcommand{\aboveverticesNR}[2][G]%
    {{#1}_{\strikethroughcommand{#2}}^{\hspace{-0.3 pt}\triangledown}}
\newcommand{\belowvertices}[2][G]%
    {{#1}^{#2}_{\hspace{-0.6 pt}\vartriangle}}
\newcommand{\belowverticesNR}[2][G]%
    {{#1}^{\strikethroughcommand{#2}}_{\hspace{-0.6 pt}\vartriangle}}
\newcommand{\lpebblingpricecompact}[1]% 
    {\formatpebblingprice{L-Peb}\bigl(#1\bigr)}

\newcommand{\scnot}[2]{#1 \langle #2 \rangle}
\newcommand{\scnotcompact}[2]{#1 \bigl\langle #2 \bigr\rangle}

\newcommand{\spcanonconfcompact}[1]%
        {\formatfunctiontosubconfiguration{canon}\bigl({#1}\bigr)}

\newcommand{\spprojsubsub}[4]%
    {\formatfunctiontosubconfiguration{proj}_{\scnot{#1}{#2}}(\scnot{#3}{#4})}
\newcommand{\spprojsubsubcompact}[4]%
    {\formatfunctiontosubconfiguration{proj}_{\scnot{#1}{#2}}%
    \bigl(\scnot{#3}{#4}\bigr)}
\newcommand{\spprojsubconf}[3]%
    {\formatfunctiontosubconfiguration{proj}_{\scnot{#1}{#2}}({#3})}
\newcommand{\spprojsubconfcompact}[3]%
    {\formatfunctiontosubconfiguration{proj}_{\scnot{#1}{#2}}\bigl({#3}\bigr)}
\newcommand{\spprojconfsub}[3]%
    {\formatfunctiontosubconfiguration{proj}_{#1}(\scnot{#2}{#3})}
\newcommand{\spprojconfsubcompact}[3]%
    {\formatfunctiontosubconfiguration{proj}_{#1}\bigl(\scnot{#2}{#3}\bigr)}
\newcommand{\spprojconfconf}[2]%
    {\formatfunctiontosubconfiguration{proj}_{#1}({#2})}
\newcommand{\spprojconfconfcompact}[2]%
    {\formatfunctiontosubconfiguration{proj}_{#1}\bigl({#2}\bigr)}

\newcommand{\spclossubcompact}[2]%
        {\formatfunctiontoset{cl}\bigl(\scnotcompact{#1}{#2}\bigr)}

\newcommand{\spintersubcompact}[2]%
        {\formatfunctiontoset{int}\bigl(\scnotcompact{#1}{#2}\bigr)}

\newcommand{\spcoversubcompact}[2]%
        {\formatfunctiontoset{cover}\bigl(\scnotcompact{#1}{#2}\bigr)}

\newcommand{\spcoverconfcompact}[1]%
        {\formatfunctiontoset{cover}\bigl({#1}\bigr)}

\newcommand{\spinducedblack}[1]%
    {\formatfunctiontoset{Bl} (#1)}
\newcommand{\spinducedwhite}[1]%
    {\formatfunctiontoset{Wh} (#1)}
\newcommand{\spinducedblackcompact}[1]%
    {\formatfunctiontoset{Bl} \bigl(#1 \bigr)}
\newcommand{\spinducedwhitecompact}[1]%
    {\formatfunctiontoset{Wh} \bigl(#1 \bigr)}

\newcommand{\pathclausedeg}[2][\pebdeg]%
    {\somenodetrueclausedeg[#1]{\vertexpath{#2}}}
\newcommand{\pathclauseNRdeg}[2][\pebdeg]%
    {\somenodetrueclausedeg[#1]{\vertexpathNR{#2}}}

\newcommand{\blacktruthdegexplicit}[4]% 
        {\setdescrcompact
        {{\textstyle \Lor_{#2 = 1}^{#3} {#1}_{#2}}}
        {{#1} \in {#4}}}

\newcommand{\binsubtree}[1]{T^{#1}}

\newcommand{\vertexpath}[1]{{P}^{#1}}
\newcommand{\vertexpathNR}[1]{{P}_{*}^{#1}}

\newcommand{\unrelatedNP}[1]%
        {T \setminus \bigl(\binsubtree{#1} \unionSP \vertexpath{#1} \bigr)}
\newcommand{\unrelatedsmallNP}[1]%
        {T \setminus (\binsubtree{#1} \unionSP \vertexpath{#1} )}

\newcommand{\pyramidgraph}[1][]{\Pi_{#1}}

\newcommand{\abovelevelblockerminsizecompact}%
    [2]{L_{\succeq{#1}}\bigl({#2}\bigr)}

\newcommand{\necessaryhidingvert}[2]%
{{#1}{\scriptstyle{\llfloor {#2} \rrfloor}}}
\newcommand{\Klawepropertyprefix}{Limited hiding-cardinality\xspace}

\newcommand{\klawepropacronym}{LHC property\xspace}

\newcommand{\nongenklaweprop}%
{non-generalized \Klawepropertyprefix property\xspace}
\newcommand{\nongenklawepropacronym}%
{non-generalized \klawepropacronym}
\newcommand{\nongenklawepropacronymWithParam}%
{(non-generalized) \klawepropacronym}

\newcommand{\siblingnonreachabiblitypropertynoref}%
{Sibling non-reachability property\xspace}
\newcommand{\Siblingnonreachabiblitypropertynoref}%
{Sibling non-reachability property\xspace}
\newcommand{\siblingnonreachabiblityproperty}%
{\siblingnonreachabiblitypropertynoref~% 
\ref{property:sibling-non-reachability-property}\xspace}
\newcommand{\Siblingnonreachabiblityproperty}%
{\Siblingnonreachabiblitypropertynoref~%
\ref{property:sibling-non-reachability-property}\xspace}

\newcommand{\introducetermanmpctext}%
    {a \introduceterm{\mpctext{}}\xspace}

\newcommand{\introducetermamultipebblingtext}%
  {a \introduceterm{\multipebblingtext{}}\xspace}

\newcommand{\blobpebblingtext}{blob-pebbling\xspace}

\newcommand{\multipebblingtext}{\blobpebblingtext}

\newcommand{\mpcostblack}[1]%
        {\formatpebblingprice{cost}_{\mpcblacks}( #1 )}
\newcommand{\mpcostwhite}[1]%
        {\formatpebblingprice{cost}_{\mpcwhites}( #1 )}

\newcommand{\blobpebblingpricecompact}[1]% 
    {\formatpebblingprice{Blob-Peb}\bigl(#1\bigr)}

\newcommand{\multipebblingpricecompact}[1]% 
    {\formatpebblingprice{Blob-Peb}\bigl(#1\bigr)}

\newcommand{\mpcblacks}{\formatfunctiontomulti{B}}
\newcommand{\mpcwhites}{\formatfunctiontomulti{W}}

\newcommand{\mpscnotcompact}[2]%
        {\big[ {#1} \big] \bigl\langle {#2} \bigr\rangle}

\newcommand{\mpctext}{\blobpebblingtext con\-fig\-u\-ra\-tion\xspace}

\newcommand{\chargeablevertices}[1]%        
{\formatfunctiontoset{chargeable}({#1}) }
\newcommand{\chargeableverticescompact}[1]%        
{\formatfunctiontoset{chargeable}\bigl({#1}\bigr) }

\newcommand{\blackschargedfor}[1][]% 
    {\mpcblacks_{#1}}
\newcommand{\whiteschargedfor}[1][]% 
    {\mpcwhites_{#1}^{\hspace{-0.3 pt}\vartriangle}}

\newcommand{\whitesbelowjustblocked}%
    {\mpcwhites_{B}^{\hspace{-0.3 pt}\vartriangle}}
\newcommand{\whitesbelowhidden}%
    {\mpcwhites_{H}^{\hspace{-0.3 pt}\vartriangle}}

\newcommand{\whitestight}%
    {\mpcwhites_{T}^{\hspace{-0.3 pt}\vartriangle}}

\DeclareFontFamily{OT1}{pzc}{}
\DeclareFontShape{OT1}{pzc}{m}{it}{<-> s * [1.200] pzcmi7t}{}
\DeclareMathAlphabet{\mathpzc}{OT1}{pzc}{m}{it}

\renewcommand{\proofsystemformat}[1]{\mathcal{#1}}

\renewcommand{\emptycl}{\bot}

\renewcommand{\refutabbrevsmall}[1]{\derivabbrevsmall{#1}{\!\bot}}
\renewcommand{\refutabbrevcompact}[1]{\derivabbrevcompact{#1}{\!\bot}}

\renewcommand{\refof}[2]{\derivof{#1}{#2}{\bot}}

\renewcommand{\extPHPnot}[2]%
    {\ephpnot{#1}{#2}}
\newcommand{\ephpnot}[2]%
    {\vphantom{\extendedversion{\formulaformat{PHP}}}
      {\smash{\extendedversion{\formulaformat{PHP}}}
        \vphantom{\formulaformat{PHP}}}^{#1}_{#2}}

\newcommand{\efphpnot}[2]%
    {\vphantom{\extendedversion{\formulaformat{FPHP}}}
      {\smash{\extendedversion{\formulaformat{FPHP}}}
        \vphantom{\formulaformat{FPHP}}}^{#1}_{#2}}

\newcommand{\ontophpnot}[2]%
    {\formulaformat{Onto}\text{-}\formulaformat{PHP}^{#1}_{#2}}

\newcommand{\ontofphpnot}[2]%
    {\formulaformat{Onto}\text{-}\formulaformat{FPHP}^{#1}_{#2}}

\renewcommand{\tseitinnot}[2]{\formulaformat{Ts}({#1,#2})}
\renewcommand{\tseitinparity}[2]{\formulaformat{PARITY}_{#1,#2}}

\newcommand{\epopnot}[1]%
    {\extendedversion{\formulaformat{POP}}_{#1}}

\newcommand{\elopnot}[1]%
    {\extendedversion{\formulaformat{LOP}}_{#1}}

\newcommand{\sinkstd}{z}

\newcommand{\Lorspecfuncproj}[4][\pebdeg]%
    {\Lor_{{#2} \in \clpospart{#3}} \specfuncwithvecarg[#1]{#4}{#2}}    
\newcommand{\Lorspecnotfuncproj}[4][\pebdeg]%
    {\Lor_{\olnot{#2} \in \clnegpart{#3}} \specnotfuncwithvecarg[#1]{#4}{#2}}
\newcommand{\witnessproj}[4][\funcpebc]%	   
    {\Lorspecfuncproj{#2}{#4}{#1} \lor \Lorspecnotfuncproj{#3}{#4}{#1}}

\newcommand{\Lorspecfuncprojnodisplay}[4][\pebdeg]%
    {\Lornodisplay_{{#2} \in \clpospart{#3}} \specfuncwithvecarg[#1]{#4}{#2}}    
\newcommand{\Lorspecnotfuncprojnodisplay}[4][\pebdeg]%
    {\Lornodisplay_{\olnot{#2} \in \clnegpart{#3}} \specnotfuncwithvecarg[#1]{#4}{#2}}
\newcommand{\witnessprojnodisplay}[4][\funcpebc]%	   
    {\Lorspecfuncprojnodisplay{#2}{#4}{#1} \lor \Lorspecnotfuncprojnodisplay{#3}{#4}{#1}}

\newcommand{\extendedpebcontrwithfunc}[3][G]%
    {{\vphantom{\formulaformat{Peb}}\smash{\extendedversion{\formulaformat{Peb}}}}^{#2}_{#1}[{#3}]}
\newcommand{\pebcontrwithfuncNT}[3][G]%
     {\dummystar\!\!\formulaformat{Peb}^{#2}_{#1}[{#3}]}
\newcommand{\genericpebcontrNT}[2][G]%
     {\pebcontrwithfuncNT[#1]{#2}{\funcpebc}}

\newcommand{\funcpebc}[1][]{f_{#1}}

\newcommand{\varvec}[1]{\vec{#1}}

\newcommand{\specfuncwithvecarg}[3][\pebdeg]%
	   {{#2}_{#1}(\varvec{#3})}
\newcommand{\specnotfuncwithvecarg}[3][\pebdeg]%
	   {\lnot {#2}_{#1}(\varvec{#3})}

\newcommand{\funcwithvecarg}[2][\pebdeg]%
           {\specfuncwithvecarg[{#1}]{\funcpebc}{#2}}
\newcommand{\notfuncwithvecarg}[2][\pebdeg]%
           {\specnotfuncwithvecarg[{#1}]{\funcpebc}{#2}}

\newcommand{\funcwithvecargstd}%
	   {\funcwithvecarg[\pebdeg]{v}}
\newcommand{\notfuncwithvecargstd}%
	   {\notfuncwithvecarg[\pebdeg]{v}}

\newcommand{\specfuncwithargspebcontr}[3][\pebdeg]%
	   {{#2}_{#1}({#3}_1, \ldots, {#3}_{\pebdeg})}
\newcommand{\specnotfuncwithargspebcontr}[3][\pebdeg]%
	   {\lnot {#2}_{#1}({#3}_1, \ldots, {#3}_{\pebdeg})}

\newcommand{\funcwithargspebcontr}[2][\pebdeg]%
	   {{\funcpebc}_{#1}({#2}_1, \ldots, {#2}_{\pebdeg})}
\newcommand{\notfuncwithargspebcontr}[2][\pebdeg]%
	   {\lnot {\funcpebc}_{#1}({#2}_1, \ldots, {#2}_{\pebdeg})}

\newcommand{\funcwithargspebcontrstd}%
	   {\funcwithargspebcontr[\pebdeg]{v}}
\newcommand{\notfuncwithargspebcontrstd}%
	   {\notfuncwithargspebcontr[\pebdeg]{v}}

\newcommand{\clpospart}[1]{#1^{+}}
\newcommand{\clnegpart}[1]{#1^{-}}
\newcommand{\fsubsttext}[1][{\funcpebc[\pebdeg]}]%
    {${#1}$-substitution\xspace}

\newcommand{\ntrues}{k}
\newcommand{\ktrue}[2][]{{\mathit{thr}}_{#1}^{#2}}
\newcommand{\ktruewithargs}[2]%
    {\ktrue[\pebdeg]{#1}({#2}_1, \ldots, {#2}_{\pebdeg})}

\newcommand{\ktruepebcontrtext}[1][\ntrues]%
    {${#1}$-true-pebbling contradiction\xspace}

\newcommand{\Xor}{{\textstyle \bigoplus}}

\newcommand{\xorvertex}[3][\pebdeg]{\Xor_{{#3}=1}^{#1} {#2}_{#3}}
\newcommand{\notxorvertex}[3][\pebdeg]{\lnot \Xor_{{#3}=1}^{#1} {#2}_{#3}}

\newcommand{\Lornotxor}[3][i]%
    {\Lor_{{#2} \in {#3}}\notxorvertex{#2}{#1}}
\newcommand{\Lorxor}[3][i]%
    {\Lor_{{#2} \in {#3}}\xorvertex{#2}{#1}}

\newcommand{\theauthorYF}{the first author\xspace}

\newcommand{\TheauthorJN}{The fourth author\xspace}

\ifthenelse{\boolean{conferenceversion}}
{
}
{

}

\ifthenelse{\boolean{conferenceversion}}
{}
{
\newtheoremstyle{metacommenttheoremstyle}% name
    {3pt}%      Space above
    {3pt}%      Space below
    {\sffamily \itshape \scriptsize
    }%         Body font
    {}%         Indent amount (empty = no indent, \parindent = para indent)
    {\bfseries \scshape \footnotesize }% Thm head font
    {:}%        Punctuation after thm head
    { }%     Space after thm head: " " = normal interword space;
    {}%         Thm head spec (can be left empty, meaning `normal')

\theoremstyle{metacommenttheoremstyle}
}
\newtheorem{jncommentcontainer}{Jakob's comment}
\newtheorem{mlcommentcontainer}{Massimo's comment}
\newtheorem{mmcommentcontainer}{Mladen's comment}
\newtheorem{mvcommentcontainer}{Marc's comment}
\newtheorem{yfcommentcontainer}{Yuval's comment}

\ifthenelse{\isundefined{\DONOTINSERTCOMMENTS}}
{ % Turn comments ON
  \usepackage[usenames,dvipsnames]{xcolor}

  \newcommand{\jncomment}[1]%
  {\begin{jncommentcontainer} \textcolor{blue}{#1} \end{jncommentcontainer}}

  \newcommand{\mlcomment}[1]%
  {\begin{mlcommentcontainer} \textcolor{OliveGreen}{#1} \end{mlcommentcontainer}}

  \newcommand{\mmcomment}[1]%
  {\begin{mmcommentcontainer} \textcolor{magenta}{#1} \end{mmcommentcontainer}}

  \newcommand{\mvcomment}[1]%
  {\begin{mvcommentcontainer} \textcolor{orange}{#1} \end{mvcommentcontainer}}

  \newcommand{\yfcomment}[1]%
  {\begin{yfcommentcontainer} \textcolor{red}{#1} \end{yfcommentcontainer}}
}
{ % Turn comments OFF
  \newcommand{\jncomment}[1]{}
  \newcommand{\mlcomment}[1]{}
  \newcommand{\mmcomment}[1]{}
  \newcommand{\mvcomment}[1]{}
  \newcommand{\yfcomment}[1]{}
}

\ifthenelse{\boolean{conferenceversion}}{
    \theoremstyle{plain}
    \newtheorem{observation}[theorem]{Observation}
    \newtheorem{proposition}[theorem]{Proposition}
}{}

\newcommand{\parens}[1]{\ensuremath{(#1)}}

\newcommand{\negconfname}[1][]{\ensuremath{\mathrm{neg}_{#1}}}
\newcommand{\negconf}[2][]{\ensuremath{\negconfname[{#1}]\parens{#2}}}

\newcommand{\negativeconf}{negated configuration\xspace}

\newcommand{\negativeconfs}{negated configurations\xspace}

\newcommand{\negativeref}{negated refutation\xspace}

\newcommand{\tseitincharge}{\ensuremath{\chi}}

\newcommand{\sspace}{\ensuremath{s}}
\newcommand{\timet}{\ensuremath{t}}
\newcommand{\widthl}{\ensuremath{\ell}}

\newcommand{\termmeasurehead}{\ensuremath{\nu}}
\newcommand{\termmeasure}[1]{\ensuremath{\termmeasurehead\parens{#1}}}

\newcommand{\confmeasurehead}{\ensuremath{\mu}}
\newcommand{\confmeasure}[1]{\ensuremath{\confmeasurehead\parens{#1}}}

\newcommand{\graphg}{\ensuremath{G}}

\newcommand{\edgesete}{\ensuremath{E}}

\newcommand{\vertexsetv}{V}
\newcommand{\vertexsetu}{U}

\newcommand{\vertexv}{\ensuremath{v}}
\newcommand{\vertexu}{\ensuremath{u}}

\newcommand{\graphdeg}{\ensuremath{d}}

\newcommand{\expansionsize}{\ensuremath{s}}
\newcommand{\expansionfactor}{\ensuremath{\delta}}

\newcommand{\edgee}{\ensuremath{e}}

\newcommand{\cutedges}[1]{\partial(#1)}

\newcommand{\tmwidth}{\ensuremath{r}}

\newcommand{\termcplxm}{term complexity measure\xspace}

\newcommand{\confcplxm}{configuration complexity measure\xspace}
\newcommand{\Confcplxm}{Configuration complexity measure\xspace}

\newcommand{\confbound}{r}

\newcommand{\expparams}{(\expansionsize, \expansionfactor)}

\newcommand{\numsubfrm}{m}

\numberwithin{equation}{section}

\begin{document}

%
% TITLE PAGE
%

\title{From Small Space to Small Width in Resolution%
  \thanks{This is a 
    slightly revised and expanded version
%        full-length version 
    of the paper
    \cite{FLMNV14FromSmallSpace}
    which appeared in
    \emph{Proceedings of the 31st Symposium on
      Theoretical Aspects of Computer Science ({STACS}~'14).}}}

\author{Yuval Filmus}
\affil{Institute for Advanced Study, Princeton, NJ, USA}
\author{Massimo Lauria}
\author{Mladen Mik\v{s}a}
\author{Jakob Nordstr\"{o}m}
\author{Marc Vinyals}
\affil{KTH Royal Institute of Technology, Stockholm, Sweden}

%    \date{\Now}
\date{\today}

\maketitle

%
% ABSTRACT
%

\begin{abstract}
  In 2003, Atserias and Dalmau resolved a major open question about
  the resolution proof system by establishing that the space
  complexity of CNF formulas is always an upper bound on the width
  needed to refute them. Their proof is beautiful but somewhat
  mysterious in that it relies heavily on tools from finite model
  theory.  We give an alternative, completely elementary proof that
  works by simple syntactic manipulations of resolution
  refutations. As a by-product, we develop a ``black-box'' technique
  for proving space lower bounds via a ``static'' complexity measure
  that works against any resolution refutation---previous techniques
  have been inherently adaptive.  We conclude by showing that the
  related question for polynomial calculus (i.e.,\ whether space is an
  upper bound on degree) seems unlikely to be resolvable by similar
  methods.
\end{abstract}

%
% PAGESTYLE AND LEFT AND RIGHT RUNNING HEADS
%

% Totally empty header and footer on first page
\thispagestyle{empty}

\pagestyle{fancy}     % results in default fancy style
% clear all header and footer fields
\fancyhead{}
\fancyfoot{}
% no rule in header or footer
\renewcommand{\headrulewidth}{0pt}
\renewcommand{\footrulewidth}{0pt}

%
% Section title centered on even and subsection centered on odd pages in header
% page number centered in footer
%    \fancyhead[CE]{\slshape \leftmark}
%    \fancyhead[CO]{\slshape \rightmark}
%
% Title centered on even pages, section centered on odd pages in header
% page number centered in footer
\fancyhead[CE]{\slshape FROM SMALL SPACE TO SMALL WIDTH IN RESOLUTION}
\fancyhead[CO]{\slshape \nouppercase{\leftmark}}
\fancyfoot[C]{\thepage}

% Increase headheight to avoid following warning message
%    ``Package Fancyhdr Warning: \headheight is too small (12.0pt):
%    Make it at least 13.59999pt.
%    We now make it that large for the rest of the document.
%    This may cause the page layout to be inconsistent, however.''
\setlength{\headheight}{13.6pt}

%
% AND HERE THE PAPER PROPER BEGINS
%

\section{Introduction}

A \introduceterm{resolution proof} for, or 
\introduceterm{resolution refutation} of, an unsatisfiable formula~$F$
in conjunctive normal form (CNF) is a 
sequence of disjunctive clauses
$(C_1, C_2, \ldots, C_\stoptime)$, where every clause~$C_t$ 
is either a member of~$F$ or
is logically implied by two previous clauses, 
%%% NEW TEXT -JN
and where the final clause is the contradictory empty clause~$\emptycl$
containing no literals. 
Resolution is arguably
the most well-studied proof system in propositional proof complexity,
and has served as a natural starting point in the quest to prove lower
bounds for increasingly stronger proof systems on 
\introduceterm{proof length/size} (which for resolution is the 
number of clauses in a proof).
% CHANGED BACK --JN
%    bounds on the \introduceterm{length/size} of proofs for increasingly
%    stronger proof systems. In particular the length of a resolution proof
%    is the number of clauses in it.
%    

Resolution is also intimately connected to SAT solving
% ,  No comma here --Jakob
in that it
lies at the foundation of state-of-the-art SAT solvers using so-called
conflict-driven clause learning (CDCL).
%    OLD TEXT
%    In this latter context one is also interested in
%    \introduceterm{proof space}.
%    
This connection has motivated the study of
\introduceterm{proof space}
as a second interesting complexity measure for resolution.
The space usage at some step~$t$ in a proof 
is measured as
%    measures 
the number of
clauses occurring before~$C_t$ that will be used to derive clauses
after~$C_t$, and the space of a proof is obtained by taking the
maximum over all steps~$t$. 

For both of these complexity measures, it turns out that a key role is
played by the auxiliary measure of \introduceterm{width}, \ie the size
of a largest clause in the proof. In a celebrated result, Ben-Sasson
and Wigderson~\cite{BW01ShortProofs} showed that there are short
resolution refutations of a formula if and only if there are also
(reasonably) narrow ones, and almost all known lower bounds on
resolution length can be (re)derived using this connection. 
In 2003, Atserias and Dalmau
%    ~%
%    \cite{AD02CombinatoricalCharacterization}
(journal version in~%
\cite{AD08CombinatoricalCharacterization})
established that width also provides lower bounds on space,
resolving a problem that had been open since the study of space
complexity of propositional proofs was initiated in the late 1990s
in~\cite{ABRW02SpaceComplexity,ET01SpaceBounds}.
% OLD TEXT
%    Again, almost all lower bounds on space can be rederived by using
%    width lower bounds and appealing 
%    to their result.  
% REPHRASING TO MAKE MEANING CLEARER
This means that for space also, almost all known lower bounds can be
rederived by using width lower bounds and appealing 
to~\cite{AD08CombinatoricalCharacterization}.
%    to their result.  
This is not a
two-way connection, however, in that formulas of almost worst-case
space complexity may require only constant width as shown in~%
\cite{BN08ShortProofs}.
% SAVING SPACE... -JN
%    (and, in particular, the strong space lower bounds in
%    that paper 
%    cannot be obtained using width).
%

%    \subsection{Our Results}
\subsection{Our Contributions}

The starting point of our work is the lower bound on space in terms of
width in~\cite{AD08CombinatoricalCharacterization}. This is a very
elegant but also magical proof in that it translates the whole problem
to Ehrenfeucht--Fraïssé games 
in
% CHANGED BACK --- 
% Google to see the difference in frequence of in vs. from... -JN
%    from 
finite model theory, and shows that
resolution space and width correspond to strategies for two opposite
players in such games. Unfortunately, this also means that one obtains
essentially no insight into what is 
happening
%%% going on 
on the proof complexity
side (other than that the bound on space in terms of width is
true). It has remained an open problem to give a more explicit, proof
complexity theoretic 
argument.
%    proof.

In this paper, we give a purely combinatorial proof in terms of simple
syntactic manipulations of resolution refutations. To summarize in one
sentence, we study the conjunctions of clauses in memory at each time
step in a small-space refutation, 
% OLD TEXT
%    negate and expand this to new sets of clauses,
%    and argue that these clauses listed in reverse order
% REPHRASING TO MAKE CLEARER   
negate these conjunctions and then expand them to conjunctive normal form
again, and finally argue that the new sets of clauses listed in
reverse order (essentially) constitute a small-width refutation of the
same formula.%
\footnote{We recently learned that a similar proof, though phrased in
  a slightly different language, was obtained independently by 
  Razborov~\cite{Razborov14personalcommunication}.}   

This new, simple proof also allows us to obtain a new technique for
proving space lower bounds. This approach is reminiscent
of~\cite{BW01ShortProofs} 
in that one defines a static ``progress measure'' on refutations and
argues that when a refutation has made substantial progress it must
have high complexity with respect to the proof complexity measure
under study. Previous 
%    space lower bounds 
lower bounds on space 
have been inherently
adaptive and in that sense less explicit. 

% REPHRASING TO MAKE CLEARER AND ADD A REFERENCE TO [BG13] -JN
%    One other important motivation for our work was the hope that a
%    simplified proof of the space-width inequality would serve as a
%    stepping stone to resolving the analogous question for the polynomial
%    calculus proof system, where the width of clauses corresponds to the
%    \introduceterm{degree} of polynomials.
One important motivation for our work was the hope that a
simplified proof of the space-width \mbox{inequality} would serve as a
stepping stone to resolving the analogous question for the polynomial
\mbox{calculus} proof system. Here the the width of clauses corresponds to the
\introduceterm{degree} of polynomials, space is measured as the total
number of monomials of all polynomials currently in memory, and the
problem is to determine whether space and degree in
polynomial calculus are related in the same way as are space
and width in resolution. A possible approach for attacking this
question was proposed in~\cite{BG13Pseudopartitions}.
% TRYING TO GET A BETTER WORDING -- JN
%    While we recently showed in \cite{FLMNV13TowardsUnderstandingPC}
%    the analogue of \cite{BN08ShortProofs}
%    that there are formulas of worst-case space complexity that require
%    only constant degree, the question of whether degree lower bounds
%    imply space lower bounds remains open. Unfortunately, as discussed
%    towards the end of this paper we show that it appears unlikely that
%    this question can be resolved by methods 
%    similar to our proof of the corresponding inequality for resolution.
In~\cite{FLMNV13TowardsUnderstandingPC} we obtained a result analogous
to~\cite{BN08ShortProofs} that there are formulas of worst-case space
complexity that require only constant degree. The question of whether
degree lower bounds imply space lower bounds remains open, however,
and other results in~\cite{FLMNV13TowardsUnderstandingPC} can be
interpreted as implying that the techniques
in~\cite{BG13Pseudopartitions} probably are not sufficient to resolve
this question. Unfortunately, as discussed towards the end of this
paper we also show that it appears unlikely that this problem can be
addressed by methods similar to our proof of the corresponding
inequality for resolution.

\subsection{Outline of This Paper}

The rest of this paper is organized as follows.
After some brief preliminaries in
\refsec{sec:prelims},
we present the new proof of the space-width inequality in~%
%%% CHANGED BACK --JN
% from
\cite{AD08CombinatoricalCharacterization} 
in
\refsec{sec:spacetowidth}.
In
\refsec{sec:space-lower-bound}
we showcase the new technique for space lower bounds by studying
so-called Tseitin formulas.
%% CHANGED BACK --JN
%    we showcase the new technique by studying space lower bounds for the
%    so-called Tseitin formulas.
\refsec{sec:spacetodegree}
explains why we believe it is unlikely that our methods will extend to
polynomial calculus. Some concluding remarks are given in
\refsec{sec:conclusion}.

\section{Preliminaries}
\label{sec:prelims}

Let us start by a brief review of the preliminaries. The following
material is standard and can be found, e.g., in  the survey~%
\cite{Nordstrom10SurveyLMCS}.

%%%
%%% The above is to avoid accusations of (self-)plagiarism.
%%% I wouldn't do this unless it had been an issue before...
%%% --Jakob
%%%

A \introduceterm{literal} over a Boolean variable $\varx$ is either
the variable $\varx$ itself (a \introduceterm{positive literal}) or
its negation that is denoted either as $\lnot \varx$ or~$\olnot{\varx}$ (a
\introduceterm{negative literal}). 
% This piece of notation is implicitly used. -JN
We define
$\olnot{\olnot{\varx}} = \varx$.
A \introduceterm{clause} 
$\clc = \lita_1 \lor \formuladots \lor \lita_{\clwidth}$ 
is a disjunction of
literals and a \introduceterm{term} $\termt = \lita_1 \land
\formuladots \land \lita_{\clwidth}$ is a conjunction of literals.
We denote the empty clause by~$\emptycl$ and the empty term
by~$\emptyset$. 
% SOMEWHAT UNCLEAR MEANING BELOW --- REPHRASING
The logical negation of a clause 
$\clc = \lita_1 \lor \formuladots \lor \lita_\clwidth$ 
is the term 
$\olnot{\lita}_1 \land \formuladots \land \olnot{\lita}_\clwidth$ 
%    $\olnot{\lita_1} \land \formuladots \land \olnot{\lita_\clwidth}$ 
that consists of the
%%% CORRECTING ENGLISH USAGE -JN
%    negations of the clause's literals, 
negations of the literals in the clause.
We will sometimes use the notation
$\lnot C $ or $ \olnot{C} $
for the term corresponding to the negation of a clause and
$ \lnot T $ or $ \olnot{T} $
for the clause negating a term.
% OLD TEXT
%    The logical negation of a clause 
%    $\clc = \lita_1 \lor \formuladots \lor \lita_\clwidth$ is the term 
%    $\olnot{\lita_1} \land \formuladots \land \olnot{\lita_\clwidth}$ 
%    that consists of the negations of the clause's literals, 
%    and vice versa. 
%    We extend the notations for logical negation to clauses and term:
%    $\lnot C $, $ \olnot{C} $, $ \lnot T $, $ \olnot{T} $.
%    
% REORDERING A BIT...
A clause (term) is \introduceterm{trivial} if it contains both a
variable and its negation. For the proof systems we study, trivial
clauses and terms can always be eliminated without any loss of
generality. 

% OLD TEXT    
%    A clause $\clc'$ \introduceterm{subsumes} clause $\clc$ %
%    if every literal from $\clc'$ also appears in $\clc$; a clause (term)
%    is \introduceterm{trivial} if it contains both a variable and its
%    negation. 

A clause $\clc'$ \introduceterm{subsumes} clause $\clc$ %
if every literal from $\clc'$ also appears in $\clc$.
A \introduceterm{$\clwidth$\nobreakdash-clause
  ($\clwidth$\nobreakdash-term}) is a clause (term) that contains at
most $\clwidth$~literals. A \introduceterm{CNF formula} $\fstd =
\clc_1 \land \formuladots \land \clc_\numsubfrm$ is a conjunction of
clauses, and a \introduceterm{DNF formula} $\fstd = \termt_1 \lor
\formuladots \lor \termt_\numsubfrm$ is a disjunction of terms.
A \introduceterm{\kcnfform{} (\mbox{$\clwidth$-DNF} formula)} is a \cnfform
(DNF formula) consisting of \xclause{\clwidth}{}s ($\clwidth$-terms).
We think of clauses, terms, and CNF formulas as sets:
the order of elements is irrelevant and there are no repetitions.

Let us next describe a slight generalization of the resolution proof
system by
Kraj{\'\i}{\v{c}}ek~\cite{K01OnTheWeak}, who introduced the family
of \introduceterm{$\tmwidth$-DNF resolution} proof systems,
%%% This notation is used below without being defined! -JN
denoted~$\resknot[\tmwidth]$,
as an intermediate step between resolution and depth-$2$ Frege
systems.   
%    Roughly speaking, for positive integers $\tmwidth$ the $\tmwidth$th
%    member of  this family
%    %    , which we denote $\resknot$, 
%    is allowed to reason in terms of \mbox{$k$-DNF} formulas.
%    For $k=1$, the lines in the proof are hence disjunctions of
%    literals, and the system
%    $\resknot[1] = \resnot$ 
%    is standard resolution. At the other extreme, 
%    $\resknot[\infty]$ 
%    is equivalent to depth-$2$ Frege.
%
% \begin{definition}[$\tmwidth$-DNF resolution \mbox{$\resknot[\tmwidth]$}]
%   \label{def:sequential-refutation}
An \introduceterm{$\tmwidth$-DNF resolution
  con\-fig\-u\-ra\-tion}~$\clsc$ is a set of $\tmwidth$-DNF
formulas. An \introduceterm{$\tmwidth$-DNF resolution refutation}
of a
CNF formula~$F$ is a sequence of configurations 
$(\clsc_0, \ldots,  \clsc_\stoptime)$ 
%    $\set{\clsc_0, \ldots,  \clsc_\stoptime}$ 
such that $\clsc_0 = \emptyset$, $\emptycl \in
\clsc_{\stoptime}$, and for $1 \leq \timet \leq \stoptime$ we obtain
$\clsc_\timet$ from $\clsc_{\timet - 1}$ by one of the following
steps:
\begin{description}
\item[Axiom download] $\clsc_\timet = \clsc_{\timet - 1} \union
  \set{\cla}$, where $A$ is a clause in $F$
  (sometimes referred to as an \introduceterm{axiom clause}).
\item[Inference] $\clsc_\timet = \clsc_{\timet - 1} \union
  \set{\cld}$, where $\cld$ is inferred by one of the following rules
  (where $G, H$ denote $\tmwidth$-DNF formulas, $\termt, \termt'$
  denote $\tmwidth$-terms, and $\lita_1, \ldots, \lita_\tmwidth$
  denote literals):
  \begin{description}
    \vspace{3pt}
  \item[$\tmwidth$-cut]
    \AxiomC{$(\lita_1 \land \formuladots \land \lita_{\tmwidth'}) \lor G$}
    \AxiomC{$\olnot{\lita}_1 \lor \formuladots \lor \olnot{\lita}_{\tmwidth'} \lor H$}
%%%    \AxiomC{$\olnot{\lita_1} \lor \formuladots \lor \olnot{\lita_{\tmwidth'}} \lor H$}
    \BinaryInfC{$G \lor H$}
    \DisplayProof,
    where $\tmwidth' \leq \tmwidth$.
    \vspace{3pt}
  \item[$\land$-introduction] 
    \AxiomC{$G \lor \termt$}
    \AxiomC{$G \lor \termt'$}
    \BinaryInfC{$G \lor (\termt \land \termt')$}
    \DisplayProof,
    as long as $\setsize{\termt \union \termt'} \leq \tmwidth$.
    \vspace{3pt}
  \item[$\land$-elimination] 
    \AxiomC{$G \lor \termt$}
    \UnaryInfC{$G \lor \termt'$}
    \DisplayProof
    for any non-empty $\termt' \subseteq \termt$.
    \vspace{3pt}
  \item[Weakening]
    \AxiomC{$G$}
    \UnaryInfC{$G \lor H$}
    \DisplayProof
    for any $\tmwidth$-DNF formula $H$.
    \vspace{3pt}
  \end{description}
\item[Erasure] $\clsc_\timet = \clsc_{\timet - 1} \setminus
  \set{\clc}$, where $\clc$ is an $\tmwidth$-DNF formula in
  $\clsc_{\timet - 1}$.
\end{description}

% \end{definition}

For
%    When  setting
$\tmwidth = 1$ 
we obtain the standard \introduceterm{resolution} proof
system. %, which we denote $\resnot$. % we never use it this notation.
In this case the only nontrivial
inference rules are weakening and $\tmwidth$-cut, 
where the former can be eliminated without loss of generality (but is
sometimes convenient to have for technical purposes) and the
latter simplifies to the \introduceterm{resolution rule}
%    which simplifies to
%    $\frac{\clc \lor \varx \quad \cld \lor \olnot{\varx}}{\clc \lor
%      \cld}$.
\begin{equation}
  \AxiomC{$\clc \lor \varx$}
  \AxiomC{$\cld \lor \olnot{\varx}$}
  \BinaryInfC{$\clc \lor \cld$}
  \DisplayProof
  \eqperiod
\end{equation}
We identify a resolution configuration~$\clsc$ with the 
CNF formula
%    conjunction
$\Land_{\clc \in \clsc} \clc$. 

The \introduceterm{length} $\lengthofarg{\proofstd}$ of an
$\tmwidth$-DNF resolution refutation $\proofstd$ is the number of
download and inference steps, and the \introduceterm{space}
$\clspaceof{\proofstd}$ 
% The proof already is in a fixed proof system, so no subindex on the
% measure -JN
%    $\clspaceof[{\resknot[\tmwidth]}]{\proofstd}$ 
is the maximal number of
$\tmwidth$-DNF formulas in any configuration in $\proofstd$. We define
the length $\lengthref[{\resknot[\tmwidth]}]{F}$ and the space
$\clspaceref[{\resknot[\tmwidth]}]{F}$ of refuting a formula~$F$ in
$\tmwidth$-DNF resolution 
% TRYING TO IMPROVE THE ENGLISH A BIT... Still not super-good, though... -JN
%    as the minimum over all refutations $F$ of
%    the respective measures.
by taking the minimum 
over all refutations $F$
with respect to the relevant measure.
%
% THE TEXT BELOW IS UNCLEAR --- which subscripts? Also the
% introduction of "clause space" as a term is a bit confusing. Trying
% to rewrite 
%    When talking about resolution refutations we drop the subscripts and
%    occasionally refer to the space as \introduceterm{clause
%      space}. Additionally, we define the \introduceterm{width}
%    $\widthofarg{\proofstd}$ of the resolution refutation $\proofstd$ to
%    be the size of a largest clause appearing in $\proofstd$, and taking
%    the minimum over all resolution refutations of $F$ we define the width
%    $\widthref{F}$ of refuting $F$.
%    
We drop the proof system
$\resknot[\tmwidth]$
from this notation when it is clear from context.

For the resolution proof system, we also define the
\introduceterm{width}
$\widthofarg{\proofstd}$ of a resolution refutation~$\proofstd$ as the
size of a largest clause in $\proofstd$, and taking 
the minimum over all resolution refutations 
%    of $F$ 
we obtain the width
$\widthref{F}$ of refuting $F$.
We remark that in the context of resolution the
space measure defined above is sometimes referred to as
\introduceterm{clause space}
to distinguish it from other space measures studied for this proof system.

\section{From Space to Width} 
%    \section{Space to width} 
\label{sec:spacetowidth}

In this section we 
present
%    give 
our new combinatorial proof that width is a
lower bound for  clause space in resolution.
The formal statement of the theorem is as follows
(in this article all CNF formulas are assumed to
be non-trivial in that they do not contain the contradictory empty
clause). 

\begin{theorem}[\cite{AD08CombinatoricalCharacterization}]
  \label{th:spaceToWidth}
  Let~$F$ be a \kcnfform{} 
%      without the empty clause 
  and let~$\refof{\proofstd}{F}$ be a
  resolution refutation in space~$\clspaceof{\proofstd} =
  \sspace$. Then there is a resolution refutation~$\proofstd'$ of~$F$
  in width~$\widthofarg{\proofstd'} \leq \sspace + \clwidth - 3$.
\end{theorem}

%    Improving the English... -JN
%    The proof idea is to take the space $\sspace$ refutation, 
The proof idea is to take the refutation $\proofstd$ in space~$\sspace$, 
% Why emphasize this particular part? -JN
%    \emph{negate  the configurations} 
negate  the configurations
one by one, 
rewrite them as equivalent sets of disjunctive clauses, and 
list these sets of clauses in reverse order.  
% OLD:   present them in reverse order. 
This forms the skeleton of the new refutation,
where all clauses have 
% OLD:   which has 
width at most $\sspace$.
% Better wording -- JN
%    That is because 
To see this, note that
each configuration in the original 
%    proof
refutation 
is the
conjunction of at most $\sspace$ clauses.
Therefore, the negation of such a configuration is a 
% OLD:    its negation is a
disjunction of at most $\sspace$ terms, which is equivalent (using
distributivity) to a conjunction of clauses of width at most~%
$\sspace$.
%
% Strange wording below --- rewriting -JN
%    The proof shows how to fill in the gaps between adjacent steps so that
%    the result is a resolution refutation. In the process, the width
%    slightly increases from $\sspace$ to $\sspace + \clwidth - 3$.
To obtain a legal resolution refutation,
we need to fill in the gaps between adjacent sets of clauses.
In this process the width
increases slightly from $\sspace$ to $\sspace + \clwidth - 3$.

% Let us start *what*? Rewriting to make this clearer. -JN
Before presenting the full proof, we need some technical results. We start
%    Let us start 
by giving a formal definition of what a \negativeconf{} is.
\begin{definition}\label{def:NegativeConfiguration}
  The \introduceterm{\negativeconf{}}~$\negconf{\clsc}$ of a
  configuration~$\clsc$ is defined by induction on the number of
  clauses in $\clsc$:
  \begin{itemize}
  \item $\negconf{\emptyset} = \set{\emptycl}$,
  \item $\negconf{\clsc \union \set{\clc}} = \setdescr{\cld \lor
      \olnot{\lita}}{\cld \in \negconf{\clsc} \text{\ and\ } \lita \in
      \clc}$,
  \end{itemize}
  where we remove trivial and subsumed clauses from the final
  configuration.
\end{definition}

\ifthenelse{\boolean{conferenceversion}}{%
\mlcomment{Observation 5 removed from conference version}
}{%

Each clause of the original configuration contributes at most one
literal to each clause of the \negativeconf. Hence, the width of
the new clauses must be small.

\begin{observation}\label{obs:WidthSpaceInMirror}
  The width of any clause in the \negativeconf{}~$\negconf{\clsc}$ is at most $|\clsc|$.
\end{observation}
}

In the proof we will use a different characterization of
\negativeconfs\ that is easier to work with.

\begin{proposition}\label{pr:AlternativeDefOfMirror}
  The \negativeconf{}~$\negconf{\clsc}$ is the set of all minimal
  (non-trivial)
  clauses~$\clc$ such that~$\lnot \clc$ implies the
  configuration~$\clsc$. That is,
  \begin{equation*}
    \negconf{\clsc} = \setdescr{\clc}{\lnot \clc \impl \clsc
      \text{\ and for every $\clc' \subseteq \clc$ it holds that
        $\lnot \clc' \nimpl \clsc$}}
    \!\eqperiod
  \end{equation*}
%      where we omit trivial clauses.
\end{proposition}

\begin{proof}
  
  Let us fix the configuration~$\clsc$ and let~$\clsd$ denote the set of
  all minimal clauses implying~$\clsc$.  
% This sounds a bit cumbersome... Plus proposition should be
% lower-case. -JN
%     clauses as in the right hand side of the equation in the Proposition.
%
  We prove that for each
  clause~$\clc \in \negconf{\clsc}$ there is a clause~$\clc' \in
  \clsd$ such that~$\clc' \subseteq \clc$ and
  % , similarly, that for each clause~$\clc \in \clsd$ there is a
  % clause~$\clc' \in \negconf{\clsc}$ such that~$\clc' \subseteq
  % \clc$
  vice versa. 
% Which result? -JN  
%      The result 
  The proposition then
  follows because 
  by definition
  neither~$\clsd$
  nor~$\negconf{\clsc}$ contains subsumed clauses.
%     by definition.

  First, let~$\clc \in \negconf{\clsc}$. By the definition
  of~$\negconf{\clsc}$ we know that for every clause~$\cld \in \clsc$
  the clause~$\clc$ contains the negation of some literal
  from~$\cld$. Hence,~$\lnot \clc$ implies~$\clsc$ as it is a
  conjunction of literals from each clause in~$\clsc$. By taking the
  minimal clause~$\clc' \subseteq \clc$ such that~$\lnot \clc' \impl
  \clsc$ we have that~$\clc' \in \clsd$.

%      In the opposite direction, let~$\clc \in \clsd$ and let us show
  In the opposite direction, we want to show for any~$\clc \in \clsd$ 
  that~$\clc$ must contain a negation of some literal in~$\cld$ for
  every clause~$\cld \in \clsc$. Assume for the sake of contradiction
  that~$\cld \in \clsc$ is a clause such that none of its literals has
  a negation
  % Just adding a small nice word -JN
  appearing
  in~$\clc$. Let~$\tvastd$ be a total 
  truth value
  assignment that
  satisfies~$\lnot \clc$
% De-emphasizing a bit... -JN
%      , which exists because~$\clc$ is not trivial. 
  (such an assignment exists because~$\clc$ is non-trivial). 
  By assumption, flipping the variables
  in~$\tvastd$ so that they falsify~$\cld$ cannot falsify~$\lnot
  \clc$. Therefore, we can find an assignment that satisfies~$\lnot
  \clc$ but falsifies~$\cld \in \clsc$, which 
% Improving the English... -JN
%    is a contradiction with
  contradicts
  the definition of~$\clsd$.  Hence, 
%      the clause 
  $\clc$ must contain a
  negation of some literal in $\cld$ for every $\cld \in \clsc$ and by
  the definition of $\negconf{\clsc}$ there is a~$\clc' \in
  \negconf{\clsc}$ such that $\clc' \subseteq \clc$.
\end{proof}

The following observation, which formalizes the main idea behind the
concept of \negativeconfs, 
% Plural here! -JN
is an immediate consequence of
\refpr{pr:AlternativeDefOfMirror}.

% In view of Massimo's comment below, I am rewriting the observation
% to talk explicitly about *clause* configurations. -JN

\begin{observation}\label{obs:MirrorIsNegationOfConfiguration}
  An assignment satisfies a 
  clause
  configuration~$\clsc$ if and
  only if it falsifies the
  negated clause configuration~%
  %     \negativeconf{}~%
  $\negconf{\clsc}$. That
  is,~$\clsc$ is logically equivalent to $\lnot \negconf{\clsc}$.
\end{observation}

%    
%    \mmcomment{The definition of negated configurations as
%      $\clsd_{\timet}$ appears for the first time after this paragraph.}
% Sort of taken care of now. -JN    

% TO ME, THE NEXT SENTENCE is kind of hanging mid-air. Trying to
% rewrite. -JN
%    The reversed sequence of \negativeconfs is not a legal resolution
%    refutation.

Recall that what we want to do is to take a 
resolution 
refutation
$\proofstd = (\clsc_0, \clsc_1, \ldots, \clsc_\stoptime)$
and argue that if $\proofstd$ has small space
complexity, then
the reversed sequence of \negativeconfs 
$\proofstd' = 
(\negconf{\clsc_\stoptime}, \negconf{\clsc_{\stoptime-1}},
\ldots, 
%    \negconf{\clsc_1}, 
\negconf{\clsc_0})$
has small width
complexity. 
However, as noted above 
$\proofstd'$ is not necessarily a legal resolution refutation.
Hence, we need to show how to derive the clauses 
% of $\clsd_{\timet}$ from
% the clauses of $\clsd_{\timet+1}$ 
in each configuration of the negated refutation
%    without a significant penalty to the width.
without increasing the width by too much.
%
% Trying to rephrase a bit to make the meaning clearer... -JN
%    The original refutation proceeds from one configuration to the next by
%    either axiom download, clause inference, or clause erasure.  We first
%    deal with inference and erasure in the following lemma, then we
%    complete the proof of \refth{th:spaceToWidth}.
%    
We do so by a case analysis over the derivation steps in the original
refutation, \ie axiom download, clause inference, and clause erasure. 
The following lemma shows that for
inference and erasure steps
all that is needed in the reverse direction is to apply weakening.

\begin{lemma}\label{lem:ADandInferenceByWeakening}
  If 
  $\clsc \impl \clsc'$, 
  then for every clause~$\clc \in \negconf{\clsc}$
  there is a clause~$\clc' \in \negconf{\clsc'}$ such that~$\clc$ is a
  weakening of~$\clc'$.
\end{lemma}

% I utterly failed to parse this proof, so I figured out what it
% wanted to say and rewrote it. -JN
%    
%    \begin{proof}
%      The clause~$\clc$ is in~$\negconf{\clsc}$ and hence by
%      \refpr{pr:AlternativeDefOfMirror} we have~$\lnot \clc \impl \clsc$
%      implying $\lnot \clc \impl \clsc'$. Hence there exists a~$\clc'
%      \subseteq \clc$ such that~$\clc' \in \negconf{\clsc'}$, again by
%      \refpr{pr:AlternativeDefOfMirror}.
%    \end{proof}
%    
\begin{proof}
  For any clause $\clc$ in~$\negconf{\clsc}$ it holds by
  \refpr{pr:AlternativeDefOfMirror} that 
  $\lnot \clc \impl \clsc$.
  Since
  $\clsc \impl \clsc'$, 
  this in turns implies that
  $\lnot \clc \impl \clsc'$. 
  Applying
  \refpr{pr:AlternativeDefOfMirror} 
  again,
  we conclude that there exists a clause
  $\clc' \subseteq \clc$ 
 such that 
 $\clc' \in \negconf{\clsc'}$.
\end{proof}

The only time in a refutation
$\proofstd =(\clsc_0, \clsc_1, \ldots, \clsc_\stoptime)$
when it does not hold that
$\clsc_{t-1} \impl \clsc_{t}$
is when an axiom clause is downloaded at time~$t$, 
and such derivation steps will require a bit more careful analysis.
We provide such an analysis in the full proof of
\refth{th:spaceToWidth}, 
which we are now ready to present.

\begin{proof}[Proof of \refth{th:spaceToWidth}]
  Let
  $\proofstd = (\clsc_0, \clsc_1, \ldots, \clsc_\stoptime)$
  be a  resolution refutation of~$F$ in space~$\sspace$. 
  For every configuration 
%      $\clsc_\timet$ in the refutation $\proofstd$, 
  $\clsc_\timet \in \proofstd$, 
  let $\clsd_\timet$ denote the corresponding
  \negativeconf{} $\negconf{\clsc_{\timet}}$.
  \ifthenelse{\boolean{conferenceversion}}{%
    By 
    the discussion preceding
    \refdef{def:NegativeConfiguration}, it is clear than each clause
    of $\clsc_{\timet}$ contributes at most one literal to each clause
    of $\clsd_{\timet}$. Hence, the clauses of $\clsd_\timet$ have
    width at most $s$.
  }{%
    By assumption each $\clsc_{\timet}$ has at most $s$ clauses, 
    and
    thus
    \refobs{obs:WidthSpaceInMirror} guarantees that $\clsd_\timet$ has
    width at most $s$.}
  We need to show how to transform the sequence
%      $\proofstd' = \set{\clsd_{\stoptime},
%        \clsd_{\stoptime - 1}, \ldots, \clsd_0}$ 
  $\proofstd' = (\clsd_{\stoptime},
    \clsd_{\stoptime - 1}, \ldots, \clsd_0)$ 
  into a 
  % We don't need a new one, just legalizing what we have will do
  % fine. -JN
  legal 
  %      new 
  resolution
  refutation of width at most~$\sspace + \clwidth - 3$.

  The initial configuration of the new refutation is $\clsd_\stoptime$
  itself, which is empty by \refdef{def:NegativeConfiguration}.
  If $\clsc_{\timet+1}$ follows $\clsc_{\timet}$ by 
%    clause 
  inference
  or 
%      clause 
  erasure, then we can derive any clause of~$\clsd_\timet$
  from a clause of~$\clsd_{\timet + 1}$ by weakening, as proven in
  \reflem{lem:ADandInferenceByWeakening}.
  If $\clsc_{\timet+1}$ follows $\clsc_{\timet}$ by axiom download,
  then we can derive~$\clsd_{\timet}$ from~$\clsd_{\timet + 1}$ in
  width at most~$\sspace + \clwidth - 3$, as we show below.
  The last configuration $\clsd_{0}$ includes the empty clause $\bot$
  by \refdef{def:NegativeConfiguration}, so the new refutation is
  complete.
  
  It remains to 
% Using "settle" here sounds a bit strange. -JN
%      settle 
  take care of
  the case of axiom download.  We claim that
  we can assume without loss of generality that
  prior to each axiom download step the space of the
  configuration~$\clsc_\timet$ is at most~$\sspace - 2$. Otherwise,
  immediately after the axiom download step the proof~$\proofstd$
  needs to erase a clause in order to maintain the space
  bound~$\sspace$. By reordering the axiom download and clause erasure
  steps we get a valid refutation of~$F$ for which it holds
  that~$\clspaceof{\clsc_\timet} \leq \sspace - 2$.

% Using "fix" here sounds a bit strange, since we cannot decide which
% axiom is being downloaded... -JN
%      Fix~
  Suppose
  $\clsc_{\timet + 1} = \clsc_\timet \union \set{\cla}$ for some
  axiom~$\cla=\lita_1 \lor \formuladots \lor
  \lita_{\widthl}$, with~$\widthl \leq \clwidth$.
  Consider now some clause~$\clc$ that is in the
  \negativeconf{}~$\clsd_\timet$ and that does not belong to~$\clsd_{\timet + 1}$.
  \ifthenelse{\boolean{conferenceversion}}{%
    Again by \refdef{def:NegativeConfiguration}, 
    the clause $ \clc $ has
    at most one literal per clause in $\clsc_\timet$, so
    $\widthofarg{\clc} \leq \sspace - 2$.
  }{%
    By \refobs{obs:WidthSpaceInMirror} $\widthofarg{\clc} \leq
    \clspaceof{\clsc_\timet} \leq \sspace - 2$.
  }%
  To derive $\clc$ from $\clsd_{\timet + 1}$ we first download axiom
  $A$ and then show how to derive~$\clc$ from the clauses
  in~$\clsd_{\timet + 1} \union \set{\cla}$.

  First, note that all clauses~$\clc_{\lita} = \clc \lor
  \olnot{\lita}$ are either contained in or are weakenings of clauses
  in~$\clsd_{\timet + 1}$. This follows easily from
  \refdef{def:NegativeConfiguration} as adding an axiom~$\cla$ to the
  configuration~$\clsc_\timet$ results in adding negations of literals
  from~$\cla$ to all clauses~$\clc \in \clsd_\timet$.  Hence, we can 
%      derive
  obtain
  $\clc$ by the following derivation:
  \begin{prooftree}
    \AxiomC{$\cla = \lita_1 \lor \formuladots \lor \lita_{\widthl}$}
    \AxiomC{$\clc_{\lita_1} = \clc \lor \olnot{\lita}_1$}
    \BinaryInfC{$\clc \lor \lita_2 \lor \formuladots \lor \lita_{\widthl}$}
    \AxiomC{$\clc_{\lita_2} = \clc \lor \olnot{\lita}_2$}
    \BinaryInfC{$\clc \lor \lita_3 \lor \formuladots \lor \lita_{\widthl}$}
    \noLine
    \UnaryInfC{$\vdots$}
    \noLine
    \UnaryInfC{$\clc \lor \lita_{\widthl}$}
    \AxiomC{$\clc_{\lita_{\widthl}} = \clc \lor \olnot{\lita}_{\widthl}$}
    \BinaryInfC{$\clc$}
  \end{prooftree}
  When~$\clc$ is the empty clause, the width of this
  derivation is upper-bounded by~$\widthofarg{\cla} \leq
  \clwidth$. Otherwise, it is upper bounded by~$\widthofarg{\clc} +
  \widthofarg{\cla} - 1 \leq \sspace + \clwidth - 3$. Any resolution
  refutation has space at least~$3$ (unless the formula contains the
  empty clause itself), so the width of~$\proofstd'$ is upper-bounded
  by~$\widthofarg{\proofstd'} \leq \sspace + \clwidth - 3$.
\end{proof}

The proof of \refth{th:spaceToWidth} also works for $\tmwidth$-DNF
resolution, with some loss in parameters. 
% NO, WE DON'T --- WE STATE THE THEOREM FIRST --JN
%    We now define the \negativeconf of an $\tmwidth$-DNF resolution
%    configuration and sketch a proof that resolution width
%    is a lower bound for $\tmwidth$-DNF resolution space.
Let us state this as a theorem and sketch the proof. 

\begin{theorem}
  Let $F$ be a \kcnfform{} and 
%      let 
  $\refof{\proofstd}{F}$ be an
  $\tmwidth$-DNF resolution refutation of $F$ in space
  $\clspaceof{\proofstd} \leq \sspace$. 
%      $\clspaceof[{\resknot[\tmwidth]}]{\proofstd} \leq \sspace$. 
%      There 
  Then there
%      is
  exists
  a resolution refutation $\proofstd'$ of $F$ in width
  at most 
  \mbox{$\widthofarg{\proofstd'} \leq (\sspace - 2) \tmwidth + \clwidth -
  1$}.
\end{theorem}

% This is a proof sketch rather than a proof. -JN
\begin{proof}[Proof sketch]
  We define the \negativeconf{}
  $\negconf[{\resknot[\tmwidth]}]{\clsc}$ of an 
  $\resknot[\tmwidth]$\nobreakdash-configuration to be
  \begin{itemize}
  \item $\negconf[{\resknot[\tmwidth]}]{\emptyset} = \set{\emptycl}$,
  \item $\negconf[{\resknot[\tmwidth]}]{\clsc \union \set{\clc}} =
    \setdescr{\cld \lor \olnot{\termt}}{\cld \in
      \negconf[{\resknot[\tmwidth]}]{\clsc} \text{\ and\ } \termt \in \clc}$,
  \end{itemize}
  with trivial and subsumed clauses removed. It is easy to see that 
  %      an $\sspace$ space %% BETTER WORDING --JN
  %      $\tmwidth$-DNF configuration 
  $\tmwidth$-DNF configuration 
  of space~$\sspace$ 
  gets transformed into a
  resolution configuration of width
  at most 
  $\sspace \tmwidth$. We can prove
  an analogue of \refpr{pr:AlternativeDefOfMirror} for this definition
  of the \negativeconf   
% BETTER WORDING --JN
%    and, hence, the analogue of
  from which the analogue of
  \reflem{lem:ADandInferenceByWeakening} 
  easily follows. The case of
  axiom download is the same as in the proof of
  \refth{th:spaceToWidth} as axioms are clauses. Hence, running the
  \negativeref backwards we get a resolution refutation of $F$ in
  width 
  at most 
  $(\sspace - 2) \tmwidth + \clwidth - 1$.
\end{proof}

\section{A Static Technique for Proving Space Lower Bounds}
%    \section{Space Lower Bounds Using Measures}
\label{sec:space-lower-bound}

%    \mmcomment{What kind of ordering should we use for references here?
%      When the references are numbered, the numbers jump around a bit.}

Looking at the proof complexity literature,
the techniques used to prove lower bounds for resolution length and
width (e.g.,
\cite{BW01ShortProofs,CS88ManyHard,H85Intractability,U87HardExamples})
% WRONG ADJECTIVE --- REPHRASING --JN
%     are essentially different
differ significantly 
from those used to prove resolution space lower bounds 
(e.g.,
\cite{ABRW02SpaceComplexity,BG03SpaceComplexity,ET01SpaceBounds})
% NO COMMA --JN   
%    ,
in that the former are \emph{static} or \emph{oblivious} while the
latter are \emph{dynamic}.

Lower bounds
on
%    for
resolution length
typically
have the following general
structure: if a refutation is too short, then we obtain a
contradiction by applying a suitable random restriction (the length of
the proof figures in by way of a union bound); so any refutation must
be long.
% Rephrasing for better English -JN
%    Lower bounds for resolution width define a complexity
%    measure, and use its properties to show that every refutation must
%    contain a complex clause; such a clause must be wide.
When proving lower bounds on resolution width, one defines a complexity
measure
%    ,  %% NO COMMA, I BELIEVE --JN
and uses the properties of this measure to show that every
refutation must contain a complex clause; in a second step one then
argues that such a complex clause must be wide.

In contrast, most lower bound proofs for resolution space use an
\emph{adversary  argument}.
Assuming that the resolution derivation 
% THIS SOUNDS STRANGE, SOMEHOW --JN
%    is in small space,
has small space,
one constructs a satisfying assignment
% Actually, ensembles of assignments are used only for PC/PCR, right?
% So I think I remove them. -JN
%    (or ensemble of assignments)
for each clause configuration.
Such assignments are updated inductively as the derivation progresses,
and one shows that the update is always possible given the assumption
that the space is small. This in turn shows that the contradictory
empty clause can never be reached, implying a space lower bound on
refutations.
The essential feature separating this kind of proofs from
the ones above is that the satisfying assignments arising during the
proof
% Well, it's not a refutation, right?
%    \emph{depend on the history of the refutation};
\emph{depend on the history of the derivation};
in contrast, the
complexity measures in width lower bounds are defined once and for
all, as are the distributions of random restrictions in length lower
bounds.
% OLD TEXT
%    In contrast, lower bound proofs for resolution space use an
%    \emph{adversary  argument}:
%    They maintain a satisfying assignment (or ensemble of
%    assignments), adjusting it from one configuration to the next. The
%    conclusion is that as long as the space is small, the empty clause can
%    never be reached. The essential feature separating these proofs from
%    the former ones is that the satisfying assignments arising during the
%    proof \emph{depend on the history of the refutation}; in contrast, the
%    complexity measures in width lower bounds are defined once and for
%    all, as is the distribution of random restrictions in length lower
%    bounds.

In this section we present a \emph{static} lower bound on resolution
space. Our proof combines the ideas of \refsec{sec:spacetowidth} and
the complexity measure for clauses used in~\cite{BW01ShortProofs}. We
define a complexity measure for configurations which can be used to
prove space lower bounds along the lines of the width lower bounds
mentioned above.

This approach works in
% Editing for language in the places below -JN
%    general; any
general in that any
complexity measure for clauses
can be transformed into a complexity measure for
%    configurations; this
configurations. This
turns many width lower bound techniques into space lower bound ones
(e.g., width lower bounds for random $3$-CNF formulas.)
In this section we give a concrete example
%    of
of this for
Tseitin formulas, which
are a family of CNFs encoding a specific type of systems of linear
equations%
\ifthenelse{\boolean{conferenceversion}}{.}{%
; see \reffig{fig:tseitin} for illustration.

\begin{figure}[tp]
  \subfigure
  [Labelled triangle graph.]	    
  {
    \label{fig:tseitin-graph}
    \begin{minipage}[b]{0.40\linewidth}
      \centering
      \includegraphics{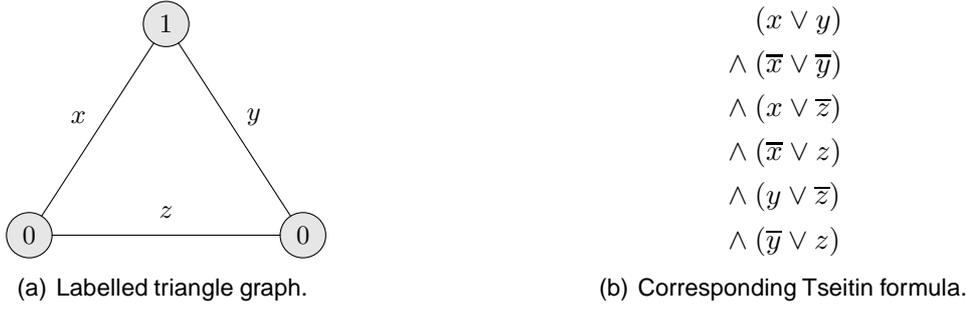}%
    \end{minipage}
  }
  \hfill
  \subfigure
  [Corresponding Tseitin formula.]	    
  {     
    \label{fig:tseitin-formula}
    \begin{minipage}[b]{0.55\linewidth}
      \centering
      \begin{gather*}
	\begin{aligned}
	  & 
	  (x \lor y)
	  \\
	  \land \ 
	  &
	  (\olnot{x} \lor \olnot{y})
	  \\
	  \land \ 
	  &
	  (x \lor \olnot{z})
	  \\
	  \land \ 
	  &
	  (\olnot{x} \lor z)
	  \\
	  \land \ 
	  &
	  (y \lor \olnot{z})
	  \\
	  \land \ 
	  &
	  (\olnot{y} \lor z)
	\end{aligned}
      \end{gather*}
    \end{minipage}
  }
  \caption{Example Tseitin formula.}
  \label{fig:tseitin}
\end{figure}

}

\begin{definition}[Tseitin formula]
  \label{def:tseitinmod2}
  Let $\graphg = (\vertexsetv, \edgesete)$ be an undirected graph and
  $\funcdescr{\tseitincharge}{\vertexsetv}{\set{0, 1}}$ be a
  function. Identify every edge $\edgee \in \edgesete$ with a variable
  $\varx_\edgee$, and let $\tseitinparity{\vertexv}{\tseitincharge}$
  denote the canonical CNF encoding of the constraint
  $\sum_{\edgee \ni \vertexv} \varx_\edgee = \tseitincharge(\vertexv) \pmod{2}$
  for any vertex $\vertexv \in \vertexsetv$. Then the
  \introduceterm{Tseitin formula} over~$\graphg$ with respect
  to~$\tseitincharge$ is $\tseitinnot{\graphg}{\tseitincharge} =
  \Land_{\vertexv \in \vertexsetv}
  \tseitinparity{\vertexv}{\tseitincharge}$.
\end{definition}

When the degree of $\graphg$ is bounded by~$\graphdeg$,
$\tseitinparity{\vertexv}{\tseitincharge}$ has at most
\mbox{$2^{\graphdeg - 1}$ clauses,} all of width at most~$\graphdeg$,
and hence $\tseitinnot{\graphg}{\tseitincharge}$ is a
\mbox{$\graphdeg$-CNF} formula with at most $2^{\graphdeg - 1}
\setsize{\vertexsetv}$ clauses. We say that a set of
vertices~$\vertexsetu$ has \introduceterm{odd (even) charge} if
$\sum_{\vertexu \in \vertexsetu} \tseitincharge(\vertexu)$ is odd
(even).
A simple parity argument shows that when $\vertices{\graphg}$ has
odd charge, $\tseitinnot{\graphg}{\tseitincharge}$ is
unsatisfiable. On the other hand, if $\graphg$ is connected then for
each $\vertexv \in \vertexsetv$ it is always possible to satisfy the
constraints $\tseitinparity{\vertexu}{\tseitincharge}$ for all
$\vertexu \neq \vertexv$.
% We don't prove this --- we *re*prove it, and with worse parameters
% to boot... -JN
%    We prove that if

The hardness of Tseitin formulas are governed by the expansion
properties of the underlying graph.

\begin{definition}[Edge expander]
%    [Edge expansion]
\label{def:edge-expansion}
  The graph $\graphg = (\vertexsetv, \edgesete)$ is an
  \introduceterm{$\expparams$-edge expander} if for every set of
  vertices $\vertexsetu \subseteq \vertexsetv$ such that
  $\setsize{\vertexsetu} \leq \expansionsize$
%    ,  %% NO COMMA -JN
%      the set of edges
%      $\cutedges{\vertexsetu} $ has size at least 
%      $\expansionfactor \setsize{\vertexsetu}$, 
  it holds that
  $\setsize{\cutedges{\vertexsetu}} \geq
  \expansionfactor \setsize{\vertexsetu}$, 
  where $\cutedges{\vertexsetu}$ is the set of
  edges of $\graphg$ with exactly one vertex in $\vertexsetu$.
\end{definition}

We next present a new technique to show that if a graph
$\graphg$ is a good edge expander, then large space is needed to refute
$\tseitinnot{\graphg}{\tseitincharge}$
in resolution.
We remark that this
%    \refth{th:TseitinSpaceLB}
was originally proven in
\cite{ABRW02SpaceComplexity,ET01SpaceBounds}
(and with slightly better parameters, as discussed below).

\begin{theorem}\label{th:TseitinSpaceLB}
  For a $\graphdeg$-regular $\expparams$-edge expander $\graphg$ it
  holds that $\clspaceof{\tseitinnot{\graphg}{\tseitincharge}} \geq
  \expansionfactor \expansionsize / \graphdeg$.
\end{theorem}

% We can't really state this theorem without attributing it properly!
% -JN

For the rest of this section we fix a particular $\graphdeg$-regular
connected graph 
$\graphg$
and
a function $\tseitincharge$ with
respect to which $\vertices{\graphg}$ has odd charge, and
consider
the corresponding Tseitin formula
$\tseitinnot{\graphg}{\tseitincharge}$.
The main tool used to prove \refth{th:TseitinSpaceLB} is a
complexity measure for configurations.
We show that if $\graphg$ is a good expander, then every refutation of
$\tseitinnot{\graphg}{\tseitincharge}$ must have a configuration with
intermediate measure.
We conclude the proof by showing that the space of a configuration is
at least 
%    the value of 
its measure
%    , %% NO COMMA, I BELIEVE --JN
if the latter falls within a specific range of
values.

We first define our \confcplxm for terms (i.e., configurations
consisting
of unit
clauses), and then 
%    we 
extend it to general configurations.
In words, the
%    The
\termcplxm is the smallest number of parity axioms of
$\tseitinnot{\graphg}{\tseitincharge}$ that collectively
contradict the
%    term. Then, the
term, and the
\confcplxm is the maximum measure
over all terms that imply the configuration.

\begin{definition}[\Confcplxm]
  The \introduceterm{\termcplxm} $\termmeasure{\termt}$ of a term
  $\termt$ is $ \termmeasure{\termt} =
  \MINOFSET{\setsize{\vertexsetv'}}{\vertexsetv' \subseteq \vertexsetv
    \text{\ and } \termt \land \Land_{\vertexv \in \vertexsetv'}
    \tseitinparity{\vertexv}{\tseitincharge} \impl \emptycl}$.

  The \introduceterm{\confcplxm} \confmeasure{\clsc} of a resolution
  configuration $\clsc$ is defined as $\confmeasure{\clsc} =
  \MAXOFSET{\termmeasure{\termt}}{\termt \impl \clsc}$. 
% THIS DOESN'T QUITE MAKE SENSE. CONTRADICTORY CONFIGURATIONS ARE
% IMPLIED BY ANYTHING. WHY SUCH A CONVOLUTED (AND STRICTLY SPEAKING
% INCORRECT) EXPLANATION?
%      When only trivial terms $\termt$ imply $\clsc$, we have
%      $\confmeasure{\clsc} = 0$. (Recall that a trivial term is a term
%      equivalent to $\bot$.)
%    
  When $\clsc$ is contradictory we have  $\confmeasure{\clsc} = 0$. 
\end{definition}

Note
that $\termmeasure{T}$ is a monotone decreasing
%    function:
function, since
$\termt \subseteq \termt'$ implies $\termmeasure{\termt}
\geq \termmeasure{\termt'}$ by definition. Hence, we only need to look
at minimal terms~$\termt$ for which $\termt \impl \clsc$ in order to
determine~$\confmeasure{\clsc}$. These minimal terms are the
\emph{negations} of the clauses in~$\negconf{\clsc}$
(compare \refpr{pr:AlternativeDefOfMirror}). 
%
% THIS IS IMPORTANT ENOUGH TO BE MADE INTO A FORMAL DEFINITION
We now introduce the convenient concept of \introduceterm{witness} for
the measure.

\begin{definition}[Witness of measure]
  A \introduceterm{witness} of the measure~$\termmeasure{\termt}$ 
%     for~$\termmeasure{\termt}$ 
  of the term~$\termt$
  is a set of
  vertices~$\vertexsetv^*$ for which $\termmeasure{\termt} =
  \setsize{\vertexsetv^*}$ and $\termt \land \Land_{\vertexv \in
    \vertexsetv^*} \tseitinparity{\vertexv}{\tseitincharge} \impl
  \emptycl$. Similarly, for configurations~$\clsc$
  %    ,  %% NO COMMA -JN
  a witness
  for~$\confmeasure{\clsc}$ is a term~$\termt^*$ for which
  $\confmeasure{\clsc} = \termmeasure{\termt^*}$ and~$\termt^* \impl
  \clsc$.
\end{definition}

%Since strengthening a term can only decrease its \termcplxm, there is always a witness for~$\confmeasure{\clsc}$ which is a minimal term implying $\clsc$, and so the set of potential witnesses is the \emph{negations} of clauses in $\negconf{\clsc}$; cf.\ \refpr{pr:AlternativeDefOfMirror}.

There is a big gap between the measure of the initial and final
configurations of a refutation, and we will see that the measure does
not change much at each step. Hence, the refutation must pass through
a configuration of intermediate measure.
Formally, if $G$ is connected then $\confmeasure{\emptyset} =
\setsize{\vertexsetv}$, because the empty term 
% I DID NOT FIND THE BELOW COMMENT VERY ENLIGHTENING -JN
%    implies $\emptyset$ and
has measure $\setsize{\vertexsetv}$, and $\confmeasure{\clsc} = 0$
when~$\emptycl \in \clsc$.
% I DID NOT FIND THE BELOW COMMENT VERY ENLIGHTENING -JN
%    , as only trivial terms imply contradiction.

% \begin{observation}
%   For a \confcplxm it holds that $\confmeasure{\emptyset} =
%   \setsize{\vertexsetv}$ and $\confmeasure{\clsc} = 0$ when $\emptycl
%   \in \clsc$.
% \end{observation}

To study how the measure changes during the refutation, we look
separately at what happens at each type of step. As in the proof of
\refth{th:spaceToWidth}, we can deal with inference and
clause erasure steps together,
whereas axiom downloads require more work.

\begin{lemma}
  If~$\clsc \impl \clsc'$ then~$\confmeasure{\clsc} \leq
  \confmeasure{\clsc'}$.
\end{lemma}
\begin{proof}
  Let~$\termt^*$ be a witness
  for~$\confmeasure{\clsc}$. Then,~$\termt^* \impl \clsc$ and, hence,
  we also have $\termt^* \impl \clsc'$. Therefore,
  \ifthenelse{\boolean{conferenceversion}}{%
    $\confmeasure{\clsc'} \geq \termmeasure{\termt^*} = \confmeasure{\clsc}$.
  }%
  {%
    $\confmeasure{\clsc'} \geq \termmeasure{\termt^*}$, because
    $\confmeasure{\clsc'}$ is equal to the maximum value of
    $\termmeasure{\termt}$ for terms $\termt$ implying $\clsc'$.
    As $\termmeasure{\termt^*}$ is equal to $\confmeasure{\clsc}$,
    the bound
    $\confmeasure{\clsc'} \geq \confmeasure{\clsc}$ follows.
  }%
\end{proof}

%    
%    Again, as in the proof of \refth{th:spaceToWidth}, axiom download
%    requires most of the work. We show that if the graph has bounded degree
%    then the measure decreases slowly.
%    

\begin{lemma}\label{lem:ADMeasureLB}
  For a clause $\cla$ in $\tseitinnot{\graphg}{\tseitincharge}$ and a
  graph $\graphg$ of bounded degree $\graphdeg$, if~$\clsc' = \clsc \union
  \set{\cla}$ then $\graphdeg \cdot \confmeasure{\clsc'} + 1 \geq
  \confmeasure{\clsc}$.
\end{lemma}

\begin{proof}
  Fix a witness $\termt^{*}$ for $\confmeasure{\clsc}$. Since
  $\confmeasure{\clsc}=\termmeasure{\termt^*}$, to prove the lemma we
  need to upper-bound the value $\termmeasure{\termt^*}$ by
  $\graphdeg \cdot \confmeasure{\clsc'}+1$.

  For any literal $\lita$ in~$\cla$, we know that $\termt^* \land
  \lita$ implies $\clsc'$ because~$\termt^*$ implies $\clsc$ and
  $\lita$ implies~$\cla$. Hence, it holds that $\confmeasure{\clsc'}
  \geq \termmeasure{\termt^* \land \lita}$, and so it will be sufficient to
  relate $\termmeasure{\termt^*}$ to the values
  $\termmeasure{\termt^* \land \lita}$.
  To this end, we look at the set of vertices $\vertexsetv^* =
  \Union_{\lita \in \cla} \vertexsetv_a \union \set{\vertexv_\cla},$
  where each $\vertexsetv_\lita$ is a witness for the corresponding
  measure~$\termmeasure{\termt^* \land \lita}$, and $\vertexv_\cla$ is
  the vertex such that $\cla \in
  \tseitinparity{\vertexv_{\cla}}{\tseitincharge}$. Note that by
  definition 
  it holds that
%      we have  %%% NEED SYMMETRY WITH "also that" BELOW -JN
  $\setsize{\vertexsetv_\lita} =
  \termmeasure{\termt^* \land \lita}$ for every $\lita \in \cla$ and
  also that 
  $\setsize{\vertexsetv^*} \leq 
  1 +  %%% AVOID AMBIGUITY HERE
  \sum_{\lita \in \cla}
  \setsize{\vertexsetv_\lita} 
%%%      + 1
  $, which
  sum
  can in turn be bounded by
  $\graphdeg \cdot \confmeasure{\clsc'} + 1$ because $\cla$ has at most
  $\graphdeg$ literals.

  We conclude the proof by showing that $\termt^* \land
  \Land_{\vertexv \in \vertexsetv^*}
  \tseitinparity{\vertexv}{\tseitincharge} \impl \emptycl$, which
  establishes
%      shows 
  that $\termmeasure{\termt^*} \leq
  \setsize{\vertexsetv^*}$. The implication holds because any
  assignment either falsifies 
  the %% ARTICLE MISSING -JN
  clause~$\cla$, and so falsifies
  $\tseitinparity{\vertexv_\cla}{\tseitincharge}$, 
%%% GRAMMAR MISMATCH --- THE ASSIGNMENT IS THE SUBJECT HERE --JN
%      or one of the
%      literals~$\lita \in \cla$
%      is satisfied. 
  or satisfies one of the
  literals~$\lita \in \cla$.
  But then we have as a subformula $\termt^*
  \land \Land_{\vertexv \in \vertexsetv_\lita}
  \tseitinparity{\vertexv}{\tseitincharge}$, which is unsatisfiable by
  the definition of $\vertexsetv_\lita$ when $\lita$ is true. The bound
  $\termmeasure{\termt^*} \leq \setsize{\vertexsetv^*}$ then follows,
  and so $\confmeasure{\clsc} \leq \setsize{\vertexsetv^*} \leq
  \graphdeg \cdot \confmeasure{\clsc'} + 1$.
\end{proof}

The preceding results imply that every resolution refutation of the
Tseitin formula has a configuration of intermediate complexity. This
holds because every refutation starts with a configuration of
measure~$\setsize{\vertexsetv}$ and needs to reach the configuration
of measure~$0$, 
as noted above,
while at each step the measure drops by a factor of at most $1
/ \graphdeg$ by 
%%% ARTICLE MISSING --- should be *the* previous lemmas -JN
%    previous lemmas.
% REPHRASING BECAUSE "previous" MAKES THIS SOUND AS IF IT WAS LONG AGO
the lemmas we just proved.
% PLUS ADD SOME CONNECTING TEXT TO MAKE CLEAR THIS IS FLUFF FOR THE
% NEXT COROLLARY... THIS WAS NOT CLEAR. -JN
Let us state this formally as a corollary.

\begin{corollary}\label{cor:BoundsOnConfMeasure}
  For any resolution refutation $\proofstd$ of a Tseitin formula
  $\tseitinnot{\graphg}{\tseitincharge}$ over a
%      \mbox{$\graphdeg$-degree}   graph $\graphg$
  connected graph $\graphg$ of bounded degree~$\graphdeg$
  and any positive integer $\confbound \leq
  \setsize{\vertexsetv}$ there exists a configuration $\clsc \in
  \proofstd$ such that the \confcplxm is bounded by
  $\confbound / \graphdeg \leq \confmeasure{\clsc} \leq
  \confbound$.
\end{corollary}

It remains to show that a configuration having intermediate measure
%    must be large (in terms of space).
must also have large space.
This part of the proof relies on the graph being an expander.
% Where is the "second" corresponding to the "first"? I don't see
% it. -JN
%    First note

%
%Also, for every minimal term $\termt$ that implies $\clsc$ it holds
%that $\clspaceof{\clsc} \geq \setsize{\termt}$: every literal of
%$\termt$ must imply at least one clause in $\clsc$ and $\clsc$ has
%$\clspaceof{\clsc}$ clauses. Let us recap all of this.

%In the case of expander graphs we have a
%space lower bound from the \confcplxm.

\begin{lemma}
  \label{lem:lemma-four-nine}
  Let $\graphg$ be an $\expparams$-edge expander graph. For every
  configuration~$\clsc$ satisfying $\confmeasure{\clsc} \leq
  \expansionsize$ it holds that~$\clspaceof{\clsc} \geq
  \expansionfactor \cdot
  \confmeasure{\clsc}$.
\end{lemma}
\begin{proof}
  To prove the lemma, we lower-bound the size of a minimal witness
  $\termt^*$ for $\confmeasure{\clsc}$ and then use the bound
  $\clspaceof{\clsc} \geq \setsize{\termt^*}$.
  %If only trivial terms
  %imply $\clsc$ then the lemma immediately follows because
  %$\confmeasure{\clsc} = 0$.
% UNNECESSARY REPETITION -- REWRITING -JN
%      The bound $\clspaceof{\clsc} \geq \setsize{\termt^*}$
  This inequality
% UNCLEAR --- REWRITING TO MAKE (HOPEFULLY) CLEARER -JN
%      follows by noting that
%      every literal of $\termt^*$ must imply at least one clause in
%      $\clsc$.
  follows by noting that at most one literal per clause in~$\clsc$ is
  needed in the implying term~$\termt^*$.

  Fix $\termt^*$ to be a minimal witness for
  $\confmeasure{\clsc}$
%    ,  %% NO COMMA -JN
  and let $\vertexsetv^*$ be a witness for
  $\termmeasure{\termt^*}$. Note that $\setsize{\vertexsetv^*} =
  \confmeasure{\clsc}$. We prove that $\termt^*$ must contain a
  variable for every edge in~$\cutedges{\vertexsetv^*}$.
  Towards contradiction, assume that
%      For the sake of contradiction, assume that
  $\termt^*$ does not
  contain some $\varx_{\edgee}$ for an edge $\edgee$ in
  $\cutedges{\vertexsetv^*}$, 
% BORDERLINE, BUT HERE I THINK THE COMMA CAN STAY AS IT CLEARLY
% INCREASES READABILITY -JN 
  and let $\vertexv_{\edgee}$ be a vertex
  in $\vertexsetv^*$ incident to $\edgee$. Let $\tvastd$ be
  an assignment that satisfies $\termt^* \land \Land_{\vertexv \in
    \vertexsetv^* \setminus \set{\vertexv_\edgee}}
  \tseitinparity{\vertexv}{\tseitincharge}$. Such an assignment must
  exist as otherwise $\vertexsetv^*$ would not be a witness for
  $\termmeasure{\termt^*}$. We can modify $\tvastd$ by changing
  the value of $\varx_{\edgee}$ so that
  $\tseitinparity{\vertexv_\edgee}{\tseitincharge}$ is satisfied. By
  the assumption, the new assignment $\tvastd'$ still satisfies
  $\termt^*$ and $\Land_{\vertexv \in \vertexsetv^* \setminus
    \set{\vertexv_\edgee}} \tseitinparity{\vertexv}{\tseitincharge}$
  as neither contains the variable $\varx_{\edgee}$. Thus, we have found an
  assignment satisfying $\termt^* \land \Land_{\vertexv \in
    \vertexsetv^*} \tseitinparity{\vertexv}{\tseitincharge}$, which is
  a contradiction.

  Hence, the term $\termt^*$ contains a variable for every edge in
  $\cutedges{\vertexsetv^*}$.
  Since $\graphg$ is an $\expparams$-edge expander and
  $\setsize{\vertexsetv^*} \leq \expansionsize$, the term $\termt^*$
  contains at least $\expansionfactor \cdot  \setsize{\vertexsetv^*}$
  variables. From $\clspaceof{\clsc} \geq \setsize{\termt^*}$ and the
  fact that $\setsize{\vertexsetv^*} = \confmeasure{\clsc}$ %
  \ifthenelse{\boolean{conferenceversion}}{%
    we prove that%
  }
  {%
    it follows that%
  }
  $\clspaceof{\clsc} \geq \expansionfactor \cdot  \confmeasure{\clsc}$ %
  \ifthenelse{\boolean{conferenceversion}}{%
    if
  }
  {%
    when
  }
  $\confmeasure{\clsc} \leq \expansionsize$.
\end{proof}

The preceding lemma and \refcor{cor:BoundsOnConfMeasure} together
imply \refth{th:TseitinSpaceLB}, because by
\refcor{cor:BoundsOnConfMeasure} there is a configuration with measure
between $\expansionsize / \graphdeg$ and $\expansionsize$, and this
configuration has space at least $\expansionfactor \expansionsize /
\graphdeg$ by 
\reflem{lem:lemma-four-nine}.
% IT'S NOT CLEAR THAT "the previous lemma" AND "the preceding lemma"
% IS THE SAME LEMMA. PLUS IT FELT LIKE GIVING A REFERENCE HERE DIDN'T
% HURT. -JN
%    the previous lemma.

We want to point out that
\refth{th:TseitinSpaceLB} gives inferior results compared to a direct
application of \refth{th:spaceToWidth} to known width lower bounds.
The bounds that we get are worse by a multiplicative factor of $1 /
\graphdeg$. 
% NO, YOU SHOW THAT "One might NOT hope"... -JN
%    One might hope to
One might have hoped to
remove this multiplicative factor by improving the bound in
\reflem{lem:ADMeasureLB}, but this is not possible because 
%% MORE VARIED LANGUAGE WILL GIVE US A SHOT AT THE NOBEL... -JN
%    that bound
this lemma
is tight.

To see this, suppose that the graph $\graphg$ is a $\graphdeg$-star: it consists of a center $\vertexv$ which is connected to $\graphdeg$ petals $\vertexu_1,\dots,\vertexu_d$ by the edges $\edgee_1,\ldots,\edgee_\graphdeg$, the charge of the center is $\tseitincharge(\vertexv) = 1$, and the charges of the petals are $\tseitincharge(\vertexu_1) = \dots = \tseitincharge(\vertexu_\graphdeg) = 0$.  Let $\cla \in \tseitinparity{\vertexv}{\tseitincharge}$ be the axiom $\cla = x_{\edgee_1} \lor \dots \lor \varx_{\edgee_\graphdeg}$. Taking $\clsc = \emptyset$ and $\clsc' = \set{\cla}$, we have that $\confmeasure{\clsc} = d+1$ while $\confmeasure{\clsc'} = 1$. The latter equality holds because every minimal term implying~$\cla$ is of the form~$\varx_{\edgee_i}$, a term which is contradicted by the single axiom $\olnot \varx_{\edgee_i} \in \tseitinparity{\vertexu_i}{\tseitincharge}$. Hence, we have an example where $\graphdeg \cdot \confmeasure{\clsc'} + 1 = \confmeasure{\clsc}$, which shows that \reflem{lem:ADMeasureLB} is tight.

% To see this, assume that the graph $\graphg$ consists of a set of
% vertices $\vertexsetv$ with one vertex~$\vertexv$ that is a neighbor
% of $\graphdeg$ disjoint subgraphs
%   $\graphg_1, \ldots, \graphg_\graphdeg$,
% each of size $(\setsize{\vertexsetv}
% - 1) / \graphdeg$. Also, let $\cla$ be one of the clauses in
% $\tseitinparity{\vertexv}{\tseitincharge}$ such that setting any
% literal
% \ifthenelse{\boolean{conferenceversion}}{}{$\lita$}
% in $\cla$ to true pushes the odd charge into one of the
% neighboring subgraphs of $\vertexv$. %
% \ifthenelse{\boolean{conferenceversion}}{}
% {%
%   That is, there is a subgraph $\graphg_i$ with a set of
%   vertices $\vertexsetv_i$ for which it holds that the formula
%   $\lita \land \Land_{\vertexv \in \vertexsetv_i}
%   \tseitinparity{\vertexv}{\tseitincharge}$ is unsatisfiable.
% }%
% Taking $\clsc = \emptyset$ and
% $\clsc' = \set{\cla}$ we have that $\confmeasure{\clsc} =
% \setsize{\vertexsetv}$ and $\confmeasure{\clsc'} =
% (\setsize{\vertexsetv} - 1) / \graphdeg$. The latter equality holds
% because every minimal term~$\termt$ satisfying $\cla$ contains exactly
% one literal from $\cla$, and so pushes the odd charge into one of
% the subgraphs neighboring $\vertexv$. This makes the vertices of that
% subgraph a witness for $\termmeasure{\termt}$. Hence, we have an
% example where $\graphdeg \cdot \confmeasure{\clsc'} + 1 =
% \confmeasure{\clsc}$, which shows that \reflem{lem:ADMeasureLB} is
% tight.

\section{From Small Space to Small Degree in Polynomial Calculus?}
\label{sec:spacetodegree}

% We know the answer to this question is "no," right? So I edited here.
% Also, the "actually" sentence sounds a bit strange.
%    An intriguing question is whether our proof of~\refth{th:spaceToWidth}
%    generalizes to the stronger algebraic proof system
%    \introduceterm{polynomial
%      calculus}~\cite{CEI96Groebner,ABRW02SpaceComplexity}. Here we
%    actually discuss the variant of polynomial calculus denoted in
%    literature as PCR.

An intriguing question is whether an analogue of the bound in
\refth{th:spaceToWidth}
holds also for the stronger algebraic proof system
\introduceterm{polynomial  calculus} introduced in~\cite{CEI96Groebner}.
In this context, it is more relevant to discuss the variant of this
system presented in
\cite{ABRW02SpaceComplexity},
known as
\introduceterm{polynomial calculus (with) resolution}
or
\introduceterm{PCR}, which we briefly describe below.

In a PCR derivation,
configurations are sets of
polynomials in $\fieldstd[x, \olnot{x}, y,
\olnot{y}, \ldots]$, where $x$ and $\olnot{x}$ are
%    formally different variables.
different formal variables.
Each polynomial~$\pcpolyp$ appearing in a configuration corresponds to
the assertion $\pcpolyp = 0$.  
The proof system contains axioms $x^2 - x$ and $x + \olnot{x} - 1$,
which restrict the values of the variables to $\{0,1\}$, and enforce
the complementarity of $x$ and $\olnot{x}$. A literal has truth value
\emph{true} if it is equal to~$0$, and truth value \emph{false} if it
is equal to~$1$. 
Each clause~$\clc$ is translated to a monomial~$\pcmonm$ with the
property that $\pcmonm = 0$ if and only if~$\clc$ is satisfied. 
For example,
the clause $x\lor y \lor \olnot{z}$
is translated to the monomial~$xy\olnot{z}$.
%
% In addition to the translations of the axiom clauses of the CNF
% formula to be refuted,
% the proof system also contain axioms
% $x^2 - x$ and $x + \olnot{x} - 1$.
% These axioms enforce that only assignments to $\set{0,1}$ are
% considered
% (and hence that all polynomials are multilinear without loss of generality)
% and that $\olnot{x}$ always takes the opposite value of~$x$.
%
%    The axioms of the system are polynomials $x^2 - x$ and $x + \olnot{x}
%    - 1$, so that only $\{0,1\}$ roots are allowed, and so that $\olnot x$
%    represents the negation of $x$.
%
%    We have
There are
two inference rules,
\introduceterm{linear combination} $\frac{p \quad q}{\alpha p + \beta
  q}$ and \introduceterm{multiplication} $\frac{p}{x p}$, where $p$
and $q$
are (previously derived) polynomials, the coefficients $\alpha, \beta$
are
elements of
%    in
$\fieldstd$, and $x$ is any variable (with or without bar).
These rules are sound in the sense that if the antecedent polynomials
evaluate to zero
under some assignment, then so does the consequent polynomial. 
%
% WHICH REFUTATION??? -JN
%    The refutation ends when~$1$ has been derived.
%    We reach a contradiction if we derive $1$.
%
A CNF formula $\fstd$ is refuted in PCR by deriving the constant
term~$1$ from the (monomials corresponding to the) clauses
in~$\fstd$.

The \introduceterm{size}, \introduceterm{degree} and
\introduceterm{monomial space} measures are
%    analogues to
analogues of
length,  width
and clause space in resolution
(counting monomials instead of clauses).
PCR
%    Polynomial calculus
can simulate
resolution refutations
% Strictly speaking, this is not really true, is it? -JN
%    with at most a
%    constant loss in these measures.
efficiently with respect to all of these measures.

% The sentence below is not very clear. Trying to edit for clarity. -JN
%    We give an example of why the proof of~\refth{th:spaceToWidth} does
%    not immediately generalize: the so-called \introduceterm{pebbling
%      contradictions}.
%

%%% WHY SMALLSKIPS? -JN
%    \smallskip

Let us now discuss why the method we use to prove
\refth{th:spaceToWidth}
is unlikely to generalize to PCR.
An example of formulas that seem hard to deal with in this way are
% BETTER WITHOUT AN ARTICLE "the" HERE -JN
%    the 
so-called
\introduceterm{pebbling contradictions},
which we 
briefly
describe next.

\begin{figure}[tp]
  \subfigure[Pyramid graph $\Pi_2$ of height 2.]          
  {
    \label{fig:pebbling-contradiction-for-Pi-2-graph}
    \begin{minipage}[b]{0.40\linewidth}
      \centering
      \includegraphics{smallPyramidHeight2.1}%
    \end{minipage}
  }
  \hfill
  \subfigure[Pebbling contradiction
    {$\pebcontr[{\pyramidgraph[2]}]{}$}.]
  {
    \label{fig:pebbling-contradiction-for-Pi-2-Peb1}
    \begin{minipage}[b]{0.55\linewidth}
      \centering
      \begin{gather*}
	\begin{aligned}
	  &
	  u
	  \\
	  \land \
	  &v
	  \\
	  \land \
	  &w
	  \\
	  \land \
	  &(\olnot{u} \lor \olnot{v} \lor x)
	  \\
	  \land \
	  &(\olnot{v} \lor \olnot{w} \lor y)
	  \\
	  \land \
	  &(\olnot{x} \lor \olnot{y} \lor z)
	  \\
	  \land \
	  &\olnot{z}
	\end{aligned}
      \end{gather*}
    \end{minipage}
  }
  \caption{Pebbling contradiction
    $\pebcontr[{\pyramidgraph[2]}]{}$
    for the pyramid graph
    $\pyramidgraph[2]$
    of height~$2$.
  }
  \label{fig:pebbling-contradiction-for-Pi-2}
\end{figure}

% Let's stick to standard terminology. Rewrote the paragraph
% below. -JN

Pebbling contradictions are defined in terms of directed acyclic
graphs (DAGs)
$\graphg=(\vertexsetv,\edgesete)$
with bounded fan-in,
where vertices with no incoming edges are
called
%    referred to as
\introduceterm{sources} and vertices without outgoing edges
%    are called
\introduceterm{sinks}.
Assume $G$ has a
unique sink~$\sinkstd$ and associate a variable $\vertexsetv$ to each vertex
$\vertexv\in\vertexsetv$. Then the
pebbling contradiction over $\graphg$ consists of the following
clauses:
\begin{itemize}
\item
  for each source vertex~$s$, a clause~$s$ (\introduceterm{source axioms}),
\item
  for each non-source vertex $\vertexv$,
  a clause
  $\Lor_{(u,v)\in\edgesete} \olnot{u} \lor \vertexv$
  (\introduceterm{pebbling axioms}),
\item
  for the sink~$\sinkstd$, a clause~$\olnot{\sinkstd}$
  (\introduceterm{sink axiom}).
\end{itemize}
%    For a directed acyclic graph $\graphg=(\vertexsetv,\edgesete)$ with a
%    unique sink~$\sinkstd$, we associate each vertex
%    $\vertexv\in\vertexsetv$ with a propositional variable $\vertexv$. The
%    pebbling contradiction consists of
%
%    one \introduceterm{vertex axiom} of the form
%    %
%    $\Lor_{(u,v)\in\edgesete} \olnot{u} \lor \vertexv$
%    %
%    for every vertex $\vertexv\in\vertexsetv$ and one unit clause
%    $\olnot{\sinkstd}$ called \introduceterm{sink axiom}.
%    %
%    When a vertex $\sourcestd$ has no incoming edges it is called a
%    \introduceterm{source} and the corresponding vertex axiom is
%    the unit clause $\sourcestd$, also called a
%    \introduceterm{source axiom}.
%
See \reffig{fig:pebbling-contradiction-for-Pi-2} for
% ARTICLE MISSING -JN
an
illustration. 
% MORE DIRECT WORDING --- SEEMS MARGINALLY BETTER -JN
%    As shown in~\cite{Ben-Sasson09SizeSpaceTradeoffs},
Ben-Sasson~\cite{Ben-Sasson09SizeSpaceTradeoffs} showed that
pebbling contradictions exhibit space-width trade-offs in resolution
in that they can always be refuted in constant width as well as in
constant space but that there are graphs for which optimizing one of
these measures necessarily causes essentially worst-case linear
behaviour for the other measure.

%
%    Pebbling contradictions over particular types of graphs exhibit
%    space-width trade-offs in resolution. They can be refuted in constant
%    space and constant width, but optimizing one of these measures
%    necessarily increases the other \cite{Ben-Sasson09SizeSpaceTradeoffs}.
%

There are two natural ways to refute pebbling contradictions in
resolution. One approach is to go ``bottom-up'' from sources to sinks in
topological order, and derive
for each vertex $\vertexv\in \vertices{\graphg}$ the unit clause~$\vertexv$
using the pebbling axiom for $v$ and the unit clauses for 
its predecessors.
% SAY "the vertex $v$" THE FIRST TIME IF YOU WANT TO SAY IT -JN
%    the predecessors of the vertex $\vertexv$. 
When the refutation reaches $z$ it derives a
contradiction with the sink axiom $\olnot{z}$. 
See
\reffig{fig:pebbling-contradiction-for-Pi-2-bottom-up}
for an example.
% IMPROVING WORDING A LITTLE BIT -JN
%    This can be done in constant width but for some graphs requires large
%    space.
This refutation can always be carried out in constant width but for
some graphs requires large space. 

The other approach is a ``top-down'' refutation
due to~\cite{Ben-Sasson09SizeSpaceTradeoffs} 
where one starts with the 
sink axiom~$\olnot{z}$ and
derives clauses of the form $\olnot{\vertexv}_{1} \lor \formuladots \lor
\olnot{\vertexv}_{\ell}$. A new clause is derived by replacing any
vertex $\vertexv_i$ in the old one by all its predecessors, i.e.,
by
resolving with the pebbling axiom for $\vertexv_i$.
Since $\graphg$ is acyclic we can repeat this process until we get to
the sources, for which the negated literals can be resolved away using
source axioms. This refutation is illustrated in
\reffig{fig:pebbling-contradiction-for-Pi-2-top-down}.
It is not hard to see that it can be performed in constant clause
space, but it might require large width.
% REWORDING, MOSTLY BECAUSE OF OTHER CHANGES ABOVE -JN
%    This refutation can be carried out in constant clause
%    space, but such a refutation might require large width.

%
%    There are two natural resolution refutations for pebbling contradictions.
%    %
%    One proceeds from the sources to the sink in topological order, and
%    for each vertex $\vertexv\in \vertices{\graphg}$ derives the
%    unit clause $\vertexv$
%    using the vertex axiom for $v$ and the unit clauses for the predecessors of
%    the vertex $\vertexv$. When the refutation reaches $z$ it derives a
%    contradiction with
%    the sink axiom $\olnot{z}$. This refutation has constant width but for
%    some graphs it requires large clause space.
%
%    The other refutation starts from the sink axiom $\olnot{z}$ and proceeds by
%    deriving clauses of the form $\olnot{\vertexv_{1}} \lor \ldots \lor
%    \olnot{\vertexv_\ell}$. We derive a new clause by replacing any
%    vertex $\vertexv_i$ in the old one by all its predecessors, i.e.,
%    resolving with vertex axiom $\vertexv_i$.
%    If $\vertexv_{i} $ is a source, we are just erasing the literal
%    $\olnot{\vertexv_{i}}$ from
%    the clause.
%    %
%    Since $\graphg$ is acyclic we can repeat this process until we get to the
%    empty clause.  This refutation can always be run in clause space
%    three, but for some graphs it requires large width.
%

%%%
%%% Type-setting refutations in plain old tabular format instead. -JN
%%%
\begin{figure}[t]
  \subfigure[Bottom-up refutation of {$\pebcontr[{\pyramidgraph[2]}]{}$}.]{
  \label{fig:pebbling-contradiction-for-Pi-2-bottom-up}
  \begin{minipage}[t]{0.45\linewidth}
    \begin{center}
      \begin{tabular}{rll}
        $1$.  & $u$ & Axiom \\     
        $2$.  & $v$ & Axiom \\
        $3$.  & $w$ & Axiom \\
        $4$.  & $\olnot{u} \lor \olnot{v} \lor x$ & Axiom \\
        $5$.  & $\olnot{v} \lor x$ & Res$(1,4)$ \\
        $6$.  & $x$ & Res$(2,5)$ \\
        $7$.  & $\olnot{v} \lor \olnot{w} \lor y$ & Axiom \\
        $8$.  & $\olnot{w} \lor y$ & Res$(2,7)$ \\
        $9$.  & $y$ & Res$(3,8)$ \\
        $10$. & $\olnot{x} \lor \olnot{y} \lor z$ & Axiom \\
        $11$. & $\olnot{y} \lor z$ & Res$(6,10)$ \\
        $12$. & $z$ & Res$(9,11)$\\
        $13$. & $\olnot{z}$ & Axiom \\
        $14$. & $\emptycl$ & Res$(12,13)$ \\
      \end{tabular}
    \end{center}
    \vspace{0.1mm}
  \end{minipage}
  }
  \hfill
  \subfigure[Top-down refutation of {$\pebcontr[{\pyramidgraph[2]}]{}$}.]{
  \label{fig:pebbling-contradiction-for-Pi-2-top-down}
   \begin{minipage}[t]{0.45\linewidth}
     \begin{center}
      \begin{tabular}{rll}
        $1$.  & $\olnot{z}$ & Axiom \\
        $2$.  & $\olnot{x} \lor \olnot{y} \lor z$ & Axiom \\
        $3$.  & $\olnot{x} \lor \olnot{y}$ & Res$(1,2)$ \\
        $4$.  & $\olnot{v} \lor \olnot{w} \lor y$ & Axiom \\
        $5$.  & $\olnot{v} \lor \olnot{w} \lor \olnot{x}$ & Res$(3,4)$ \\
        $6$.  & $\olnot{u} \lor \olnot{v} \lor x$ & Axiom \\
        $7$.  & $\olnot{u} \lor \olnot{v} \lor \olnot{w}$ & Res$(5,6)$ \\
        $8$.  & $w$ & Axiom \\
        $9$.  & $\olnot{u} \lor \olnot{v}$ & Res$(7,8)$ \\
        $10$. & $v$ & Axiom \\
        $11$. & $\olnot{u}$ & Res$(9,10)$ \\
        $12$. & $u$ & Axiom \\
        $13$. & $\emptycl$ & Res$(11,12)$ \\
        \vphantom{$14$.} & \mbox { }  \\  %% Trying to get spacing right
      \end{tabular}
     \end{center}
    \vspace{0.1mm}
   \end{minipage}
  }
  \caption{Example resolution refutations of
    pebbling contradiction
    $\pebcontr[{\pyramidgraph[2]}]{}$.
  }
  \label{fig:pebbling-contradiction-for-Pi-2-refutations}
\end{figure}

%    \input{figRefutations.tex}

% REWORDING SINCE THIS IS ARGUABLY NOT SUPER-OBVIOUS -JN
%    Now, one can observe that 
A careful study now reveals that
the transformation of configurations in our
proof of \refth{th:spaceToWidth} maps either of the two refutations 
describe above
into the other one. Instead of providing a formal argument, 
% WHAT DOES IT MEAN TO "consult the refutations"? -JN
%    we suggest consulting the two refutations
%    in~\reffig{fig:pebbling-contradiction-for-Pi-2-refutations},
we encourage the reader to compute the tranformations of the refutations in
\reftwofigs{fig:pebbling-contradiction-for-Pi-2-bottom-up}{fig:pebbling-contradiction-for-Pi-2-top-down},
observing that the axioms are downloaded in opposite order in the two
derivations. 
%    observing the order in which axioms are downloaded in both refutations. 
This observation
is the main reason why our proof does not seem to generalize to PCR,
as we now explain. 

In PCR, we can represent any conjunction
of literals $\lita_1 \land \formuladots \land \lita_{\tmwidth}$ as the
binomial $1 - \prod_i \olnot{\lita}_i$. Using this encoding with the
bottom-up approach yields a third refutation, which has constant space but
possibly large degree: the fact that a set of vertices $\vertexsetu$
``are true'' can be stored as the high-degree binomial $1 -
\prod_{\vertexv \in \vertexsetu} \olnot\vertexv$ instead of as a
collection of low-degree monomials $\setdescr{\vertexv}{\vertexv \in
  \vertexsetu}$. 
Hence, there are constant space PCR refutations of pebbling
contradictions in both the 
bottom-up and the top-down 
% SHOULD BE PLURAL, I BELIEVE -JN
%    direction
directions.  This in turn means that
if our proof method were to work for PCR, we would need to find
constant degree refutations in both directions. For the top-down case
it seems unlikely that such a refutation exists,
since 
clauses of the form
%%% HERE V' LOOKS A BIT like a comma in the text, so U seems better -JN
$\bigvee_{\vertexv \in \vertexsetu} \olnot\vertexv$ cannot be
represented as low-degree polynomials. 

%
%    The transformation of configurations in \refth{th:spaceToWidth} maps
%    either of these proofs into the other.
%    %
%    This is the main reason why our proof does not seem to generalize to
%    polynomial calculus: we can represent any conjunction of literals
%    $\lita_1 \land \formuladots \land \lita_{\tmwidth}$ as the binomial $1
%    - \prod_i \olnot{\lita}_i$. Using this encoding with the
%    bottom-up proof yields a third proof, which has constant space but
%    possibly large degree. Hence, there are constant space polynomial
%    calculus refutations
%    of pebbling contradictions in both directions.
%    %
%    %    %    This means that for our method to work for polynomial
%    calculus
%    we would need
%    constant degree refutations in both directions. For the top-down case
%    it seems unlikely that such a refutation exists.
%

\section{Concluding Remarks}
\label{sec:conclusion}

% Rephrasing this again --- Atserias-Dalmau is also completely
% explicit, right? -JN  
%    In this work, we present an alternative, completely explicit, proof
%    of the result by Atserias and Dalmau~%
In this work, we present an alternative, completely elementary, proof
of the result by Atserias and 
\mbox{Dalmau~\cite{AD08CombinatoricalCharacterization}}
that space is an upper bound on width in resolution.
Our construction gives a syntactic way to convert a small-space
resolution refutation into a refutation in small width.
%
% Rewrote above taking the suggestion into consideration --JN
%
% A VARIANT FOR YOUR CONSIDERATION
%
% In this work we give an alternative---completely explicit---proof of
% the result by Atserias and Dalmau~%
% \cite{AD08CombinatoricalCharacterization} 
% that in resolution the clause space is an upper bound on width. Our
% proof shows how to convert a small-space resolution refutation into
% a refutation in small width.
% 
%
We also exhibit a new ``black-box'' approach for proving space
lower bounds that works by defining a progress measure 
%    for configurations in resolution refutations 
\mbox{à la} Ben-Sasson and Wigderson~%
\cite{BW01ShortProofs}
and showing that when a refutation has made medium progress towards a
contradiction it must be using a lot of space.  We believe that these
techniques shed interesting new light on resolution space complexity
and hope that they will serve 
to increase our understanding of this
notoriously tricky complexity measure.

As an example of a question about resolution space that still remains open, 
suppose we are given a \mbox{$k$-CNF} formula that is 
guaranteed to be refutable in constant space. 
%    consider a \mbox{$k$-CNF} formula that is refutable in constant space. 
By \cite{AD08CombinatoricalCharacterization} it is also
refutable in constant width, and a simple counting argument then shows
that exhaustive search 
%    for refutations 
in small width will find a polynomial-length resolution refutation. 
But is there any way of obtaining such a short refutation from a
refutation in small space that is more explicit than doing exhaustive
search?  
%    
%    But is there any way we can find such a short refutation from a
%    small-space refutation other than by exhaustive search? 
%    
And can we obtain a short refutation without blowing up the space by
more than, say, a constant factor?
Known length-space trade-off results for resolution in
\ifthenelse{\boolean{conferenceversion}}
{\cite{BBI12TimeSpace,BNT12SomeTradeoffs,BN11UnderstandingSpace,Nordstrom09SimplifiedWay}}
{\cite{BBI12TimeSpace,BN11UnderstandingSpace,BNT12SomeTradeoffs,Nordstrom09SimplifiedWay}}
do not answer this question as they do not apply to this range of
parameters. 
Unfortunately, our new proof of the space-width inequality cannot be
used to resolve this question either, since in the worst case the
resolution refutation we obtain might be as bad as the one found by
exhaustive search of small-width refutations 
(or even worse, due to repetition of clauses). 
This would seem to be
inherent---a recent result~\cite{ALN13} shows that there are formulas
refutable in space and width~$s$ where the shortest refutation has
length~$n^{\bigomega{s}}$, \ie matching the exhaustive search upper
bound up to a (small) constant factor in the exponent.

An even more intriguing question is how the space and degree measures
are related in polynomial calculus, as discussed in
\refsec{sec:spacetodegree}.
For most relations between length, space, and width in resolution, it
turns out that they carry over with little or no modification to size,
space, and degree, respectively, in polynomial calculus. So can it be
that it also 
holds that space yields upper bounds on degree in polynomial calculus?
Or could perhaps even the stronger claim hold that polynomial
calculus space is an upper bound on resolution width? These questions
remain wide open, but in the recent paper~% 
\cite{FLMNV13TowardsUnderstandingPC}
we made some limited progress by showing that if a formula requires
large resolution width, then the ``XORified version'' of the formula
requires large polynomial calculus space. We refer to the introductory
section of~%  
\cite{FLMNV13TowardsUnderstandingPC}
for a more detailed discussion of these issues.

\section*{Acknowledgments}

The authors wish to thank
Albert Atserias, 
Ilario Bonacina,
Nicola Galesi,
and 
Li-Yang Tan
for stimulating discussions on 
topics
%    themes 
related to this work.
We would also like to thank Alexander Razborov for sharing his proof
of the theorem that space upper-bounds width, which is very similar to
ours although expressed in a different language.

The research of \theauthorYF
has received funding from the
European Union's  Seventh Framework Programme (FP7/2007--2013) 
under grant agreement no.~238381.
%
%    , and part of his work
Part of the work of \theauthorYF 
was performed while at the University
of Toronto and while visiting KTH Royal Institute of Technology.
The other authors were funded by the
European Research Council under the European Union's Seventh Framework
Programme \mbox{(FP7/2007--2013) /} ERC grant agreement no.~279611.
\TheauthorJN 
was also supported by
Swedish Research Council grants 
\mbox{621-2010-4797}
and
\mbox{621-2012-5645}.

%
% BIBLIOGRAPHY
%

\bibliography{refArticlesUTF8,refBooksUTF8,refOtherUTF8}

%    \bibliographystyle{plain}   % standard BibTeX (numbers as labels)
%%% Change to jabbrv_abbrv to add notes and months back
%%% Change to abbrv to add full journal names back
% \bibliographystyle{jabbrv_abbrv_nonotes}
%    \bibliographystyle{abbrv}   % Like standard, but more compact entries
\bibliographystyle{alpha}   % Labels [ABRW02] and similar

\end{document}